\documentclass{article}

\usepackage{times} 
\usepackage{xspace}
\usepackage{algorithmic}
\usepackage{algorithm}
\usepackage{amssymb,amsmath,amsthm}
\usepackage{pstricks,pst-plot}
\usepackage{color}
\usepackage{typearea}

\typearea{11}

\newcommand{\Ex}[1]{\mathbb{E}\left[#1\right]}
\renewcommand{\Pr}[2][]{\mathbb{P}_{#1}\left[#2\right]}
\newcommand{\lat}{\ell}
\newcommand{\lavg}{L_{\text{av}}}
\renewcommand{\L}{\mathcal{L}}
\newcommand{\paths}{\mathcal{P}}

\newcommand{\e}{\mathrm{e}}
\newcommand{\expf}[1]{\e^{#1}}
\newcommand{\Oh}[1]{\mathcal{O}\left(#1\right)}

\newcommand{\DX}{\Delta\tilde x_e}
\newcommand{\EDX}{\Ex{\DX}}
\newcommand{\Dl}{\Delta\tilde\lat_e}

\newcommand{\IP}{\textsc{Imitation Protocol}\xspace}
\newcommand{\XP}{\textsc{Exploration Protocol}\xspace}
\newcommand{\PLS}{\textsf{PLS}\xspace}

\newcommand{\etal}{{et al.}\xspace}
\newcommand{\ie}{{i.\,e.}\xspace}

\newcommand{\calN}{\mathcal{N}}

\newcommand{\calR}{\mathcal{R}}

\newcommand{\Rin}{\mathcal{R}_{\mbox{\tiny in}}}
\newcommand{\Rout}{\mathcal{R}_{\mbox{\tiny out}}}
\newcommand{\Sin}{S^{\mbox{\tiny in}}}
\newcommand{\Sout}{S^{\mbox{\tiny out}}}
\newcommand{\Sinit}{S^{\mbox{\tiny init}}}

\newcommand{\NN}{\ensuremath{\mathbb{N}}}

\newcommand{\RR}{\ensuremath{\mathbb{R}}}

\newtheorem{definition}{Definition}
\newtheorem{theorem}{Theorem}
\newtheorem{corollary}[theorem]{Corollary}
\newtheorem{lemma}[theorem]{Lemma}
\newtheorem{fact}[theorem]{Fact}

\floatname{algorithm}{Protocol}

\DeclareMathOperator{\opt}{opt}
\DeclareMathOperator{\poly}{poly}

\title{Concurrent Imitation Dynamics in Congestion Games\thanks{This work was in part supported by the DFG through German UMIC-excellence cluster at RWTH Aachen University.}}

\author{
  Heiner Ackermann \\ RWTH Aachen University \\ ackermann@cs.rwth-aachen.de \and
  Petra Berenbrink\thanks{Supported by an NSERC grant. Part of this work was done while author visited RWTH Aachen University.} \\ Simon Fraser University \\ petra@cs.sfu.ca \and
  Simon Fischer \\ RWTH Aachen University \\ fischer@cs.rwth-aachen.de \and
  Martin Hoefer\thanks{Supported by the German Academic Exchange Service (DAAD) within the PostDoc-Program.} \\ Stanford University \\ mhoefer@cs.rwth-aachen.de
}

\begin{document}

\maketitle

  \begin{abstract}
  Imitating successful behavior is a natural and frequently applied
  approach to trust in when facing scenarios for which we have little
  or no experience upon which we can base our decision. In this paper,
  we consider such behavior in atomic congestion games. We propose to
  study concurrent imitation dynamics that emerge when each player
  samples another player and possibly imitates this agents' strategy
  if the anticipated latency gain is sufficiently large. Our main
  focus is on convergence properties. Using a potential function
  argument, we show that our dynamics converge in a monotonic fashion
  to stable states. In such a state none of the players can improve
  its latency by imitating somebody else.

  As our main result, we show rapid convergence to approximate
  equilibria. At an approximate equilibrium only a small fraction of
  agents sustains a latency significantly above or below average. In
  particular, imitation dynamics behave like fully polynomial time
  approximation schemes (FPTAS). Fixing all other parameters, the
  convergence time depends only in a logarithmic fashion on the number
  of agents.
  

  Since imitation processes are not innovative they cannot discover
  unused strategies. Furthermore, strategies may become extinct with
  non-zero probability. For the case of singleton games, we show that
  the probability of this event occurring is negligible. Additionally,
  we prove that the social cost of a stable state reached by our
  dynamics is not much worse than an optimal state in singleton
  congestion games with linear latency function. Finally, we discuss
  how the protocol can be extended such that, in the long run,
  dynamics converge to a Nash equilibrium.
\end{abstract}


  \section{Introduction}

We study imitation dynamics that emerge if myopic players concurrently
imitate each other in order to improve on their own situation.  In
scenarios for which players have little or no experience upon which
they can base their decisions, or in which precise knowledge about the
available options and their consequences is absent, it is a good
strategy to \emph{imitate} successful behavior. Thus, it is not
surprising that such imitating behavior can frequently be observed,
and has already been studied intensively in economics and game
theory~\cite{Hofbauer/Sigmund:EGT:98,Weibull:EGT:95}.

We analyze such imitation dynamics in the context of symmetric
congestion games~\cite{Rosenthal:PureNash:73}. As an example of such a
game consider a network congestion game in which players strive to
allocate paths with minimum latency between the same source-sink pair
in a network. The latency of a path equals the sum of the latencies
of the edges in that path and the latency of an edge depends on the
number of players sharing it.

We consider a simple imitation rule according to which players strive
to improve their individual latencies over time by imitating others in
a concurrent and round-based fashion.  This \IP has several appealing
properties: it is simple, stateless, based on local information, and
is compatible with the selfish incentives of the players. The \IP
consists of a sampling and a migration step. First, each player
samples another player uniformly at random. Then he considers the
latency gain that he would have by adopting the strategy of the
sampled player, under the assumption that no-one else changes his
strategy. If this latency gain is not too small our player adopts the
sampled strategy with a \emph{migration probability} mainly depending
on the anticipated latency gain. The major technical challenge in
designing such a concurrent protocol is to avoid \emph{overshooting
  effects}. Overshooting occurs if too many players sample other
players currently using the same strategy, and if all of them migrate
towards it. In this case their latency might be greater than before
the migration. In order to avoid overshooting, the migration
probabilities
have to be defined appropriately without sacrificing the benefit of
concurrency. We propose to scale the migration probabilities by the
\emph{elasticity} of the latency functions in order to avoid
overshooting. The elasticity of a function at point $x$ describes the
proportional growth of the function value as a result of a
proportional growth of its argument. Note that in case of polynomial
latency functions with positive coefficients and maximum degree $d$
the elasticity is upper bounded by $d$.

A natural solution concept in this scenario is imitation-stability. A
state is \emph{imitation-stable} if no more improvements are possible
based on the \IP. We analyse convergence properties with respect to
this solution concept.

\subsection{Our Results}

As our first result we prove that the \IP succeeds in avoiding
overshooting effects and converges in a monotonic fashion
(Section~\ref{Sec:ExactConvergence}). More precisely, we show that a
well-known potential function (Rosenthal~\cite{Rosenthal:PureNash:73})
decreases on expectation as long as the system is not yet at an
imitation-stable state. Thus, the potential is a
\emph{super-martingale} and eventually reaches a local minimum,
corresponding to an imitation-stable state. Hence, as a corollary, we
see that an imitation-stable state is reached in pseudopolynomial
time.

Our main result, presented in Section~\ref{Section:ApproxStable},
however, is a much stronger bound on the time to reach approximate
imitation-stable states. What is a natural definition of approximately
stable states in our setting? By repeatedly sampling other agents, an
agent gets to know the average latency of the system. It is 
approximately satisfied, if it does not sustain a latency much larger
than the average. Hence, we say that a state is approximately stable if
almost all agents are almost satisfied. More precisely, we consider
states in which at most a $\delta$-fraction of the agents deviates by
more that an $\epsilon$-fraction (in any direction) from the average
latency. We show that the expected time to reach such a state is
polynomial in the inverse of the approximation parameters $\delta$ and
$\epsilon$ as well as in the maximum elasticity of the latency
functions, and logarithmic in the ratio between maximum and minimum
potential. Hence, if the maximum latency of a path is fixed, the time
is only logarithmic in the number of players and independent of the
size of the strategy space and the number of resources.

We complement these results by various lower bounds. First, it is
clear that pseudopolynomial time is required to reach exact
imitation-stable states. This follows from the fact that there exist
states in which all latency improvements are arbitrarily small,
resulting in arbitrarily small migration probabilities. Hence, already
a single step may take pseudopolynomially long. As a concept of
approximate stable states one could have required \emph{all} agents to
be approximately satisfied, rather than only all but a
$\delta$-fraction.  This, however, would require to wait a polynomial
number of rounds for the last agent to become approximately satisfied,
as opposed to our logarithmic bound. Finally, we consider sequential
imitation processes in which only one agent may move at a time. We
extend a construction from~\cite{Ackermann/etal:Combinatorial:06} to
show that there exist instances in which the shortest sequence of
imitations that leads to an imitation-stable state is exponentially
long.
 
The \IP has one drawback: It is not innovative in the following sense.
It might happen with small but non-zero probability that all players
currently using the same strategy $P$ migrate towards other strategies
and no other player migrates towards $P$. In this case, the knowledge
about the existence of strategy $P$ is lost and cannot be regained. For
singleton games, \ie, games in which each strategy is a singleton set,
in which empty links have latency zero, we show in
Section~\ref{sec:singleton} that the probability of this event
occurring in a polynomial number of rounds is negligible. This also has
an important consequence: The cost of a state to which the \IP
converges is, on expectation, not much worse than the cost of a Nash
equilibrium. More precisely, we show for the case of linear latency
functions that the expected cost of a state to which the \IP converges
is within a constant factor of the optimal solution.

We conclude with a discussion of a possible extension of the \IP in
Section~\ref{sec:exploration}. In cases, in which convergence to a
Nash equilibrium is required, it is possible to adjust the dynamics
and occasionally let players use an \XP. Using such a protocol,
players sample other strategies directly instead of sampling them by
looking at other players. We show that a suitable definition of such a
protocol and a suitable combination with the \IP guarantee convergence
to Nash equilibria in the long run.

To the best of our knowledge, this is the first work that considers
concurrent protocols for atomic congestion games that are not restricted
to parallel links and linear latency functions.

\subsection{Related Work}

Rosenthal~\cite{Rosenthal:PureNash:73} proves that every congestion
game possesses a Nash equilibrium, and that better response dynamics
converge to Nash equilibria. In these dynamics players have complete
knowledge, and, in every round, only a single player deviates to a
better strategy than it currently uses. Fabrikant
\etal~\cite{Fabrikant/etal:ComplexityPure:04}, however, observe that,
in general, from an appropriately chosen initial state it takes
exponentially many steps until players finally reach an equilibrium.
This negative result still holds in games with $\epsilon$-greedy
players, \ie, in games in which players only deviate if their latency
decreases by a relative factor of at least
$1+\epsilon$~\cite{Ackermann/etal:Combinatorial:06,
  Chien/Sinclair:Convergence:07,
  Skopalik/Voecking:Inapproximability:08}.  Moreover, Fabrikant
\etal~\cite{Fabrikant/etal:ComplexityPure:04} prove that, in general,
computing a Nash equilibrium is \PLS-complete. Their result still
holds in the case of asymmetric network congestion games.  In
addition, Skopalik and
V\"ocking~\cite{Skopalik/Voecking:Inapproximability:08} prove that
even computing an approximate Nash equilibrium is \PLS-complete. On
the positive side, best response dynamics converge quickly in
singleton and matroid congestion
games~\cite{Ackermann/etal:Combinatorial:06,Ieong/etal:FastCompact:05}.
Additionally, Chien and Sinclair~\cite{Chien/Sinclair:Convergence:07}
consider the convergence time of best response dynamics to approximate
Nash equilibria in symmetric games. They prove fast convergence to
approximate Nash equilibria provided that the latency of a resource
increases by at most a factor for each additional user. Finally,
Goldberg~\cite{Goldberg:RLSLoadBalancing:04} considers a protocol
applied to a scenario where $n$ weighted users assign load to $m$
parallel links and the latency equals the load of a resource. In this
protocol, randomly selected players move sequentially, and migrate to
a randomly selected resource if this improves their latency. The
expected time to reach a Nash equilibrium is pseudopolynomial. Results
considering other protocols and links with latency functions are
presented in~\cite{EvenDar/etal:ConvergenceTime:03} 

The social cost of (approximate) Nash equilibria in congestion games
has been subject to numerous studies. The most prominent concept has
been the price of anarchy~\cite{Koutsoupias/Papadimitriou:Worst:99},
which is the ratio of the worst cost of any Nash equilibrium over the
cost of an optimal assignment. Roughgarden and
Tardos~\cite{Roughgarden/Tardos:HowBad:02} conducted the first study of
general, non-atomic congestion games and showed a tight bound of 4/3
for the price of anarchy with linear latency functions. For atomic
games and linear latencies, Awerbuch
\etal~\cite{Awerbuch/etal:PriceUnsplittable:05} and Christodoulou and
Koutsoupias~\cite{Christodoulou/Koutsoupias:AnarchyFinite:05} show a
tight bound of 2.5. The special case of (weighted) singleton games has
been of particularly strong interest, and we refer the reader
to~\cite[chapter 20]{Nisan:AGT:07} for an introduction to the numerous
results. In terms of dynamics, Awerbuch
\etal~\cite{Awerbuch/etal:FastConvergence:08} consider the number of
best-response steps required to reach a desirable state, which has a
social cost only a constant factor larger than that of a social
optimum. They show that even in congestion games with linear latencies
there are exponentially long best-response sequences for reaching such
a desirable state. In contrast, Fanelli
\etal~\cite{Fanelli/etal:FastConvergence:08} show that for linear
latency functions there are also much faster best response sequences
that reach a desirable state after at most $\Theta(n \log\log n)$
steps.

Recently, concurrent protocols have been studied in various models and
under various assumptions. Even-Dar and
Mansour~\cite{EvenDar/Mansour:FastConvergence:05} consider concurrent
protocols in a setting where the links have speeds. However, their
protocols require global knowledge in the sense that the users must be
able to determine the set of underloaded and overloaded links. Given
this knowledge, the convergence time is doubly logarithmic in the
number of players. In~\cite{Berenbrink/etal:LoadBalancing:06} the
authors consider a distributed protocol for the case that the latency
equals the load that does not rely on this knowledge. Their bounds on
the convergence time are also doubly logarithmic in the number of
players but polynomial in the number of links.
In~\cite{Berenbrink/etal:ConvergenceWeighted:07} the results are
generalized to the case of weighted jobs. In this case, the
convergence time is only pseudopolynomial, \ie, polynomial in the
number of users, links, and in the maximum weight. Finally, Fotakis
\etal~\cite{Fotakis/etal:CGConvergence:08} consider a scenario with
latency functions for every resource. Their protocol involves local
coordination among the players sharing a resource. For the family of
games in which the number of players asymptotically equals the number
of resources they prove fast convergence to almost Nash equilibria.
Intuitively, an almost Nash equilibrium is a state in which there are
not too many too expensive and too cheap resources.
In~\cite{Fischer/etal:LocalInfo:08}, a load balancing scenario is
considered in which no information about the target resource is
available. The authors present an efficient protocol in which the
migration probability depends purely on the cost of the currently
selected strategy.

In~\cite{Fischer/etal:FastConvergence:06} the authors consider
congestion games in the Wardrop model, where an infinite population of
players carries an infinitesimal amount of load each. They consider a
protocol similar to ours and prove that with respect to approximate
equilibria it behaves like an FPTAS, \ie, it reaches an approximate
equilibrium in time polynomial in the approximation parameters and the
representation length of the instance (e.\,g., if the latency functions
are polynomials in coefficient representation). In contrast to our work
the analysis of the continuous model does not have to take into account
probabilistic effects.

Our protocol is based on the notion of imitation, a concept frequently
applied in evolutionary game theory. For an introduction to imitation
dynamics, see, e\,g., \cite{Hofbauer/Sigmund:EGT:98,Weibull:EGT:95}.

\section{Congestion Games and Imitation Dynamics}

In this section, we provide a formal description of our model. We
define congestion games in terms of networks, that is, the strategy
space of each player corresponds to the set of paths connecting a
particular source-sink pair in a network. However, our results are
independent of this definition and still hold in general, symmetric
congestion games. Furthermore, we introduce the slope and the
elasticity of latency functions, and give a precise definition of the
\IP.

\subsection{Symmetric Network Congestion Games} 

A symmetric network congestion game is a tuple $(G, (s,t),
\calN,(\lat_e)_{e\in E})$, where $G=(V,E)$ denotes a network with
vertices $V$ and $m$ directed edges $E$, and $s\in V$ and $t\in V$
denote a source and a sink vertex. Furthermore, $\calN$ denotes a set
of $n$ \emph{agents} or \emph{players}, and $(\lat_e)_{e\in E}$ a
family of non-decreasing and differentiable latency functions $\lat_e
\colon \RR_{\ge0} \rightarrow \RR_{\ge0}$. We assume that for all $e\in
E$, the latency functions satisfy $\lat_e(x)>0$ for all $x>0$. The
strategy space of all players equals the set of paths $\paths$
connecting the source $s$ with the sink $t$. If $G$ consists of two
nodes $s$ and $t$ only, which are connected by a set of parallel links,
then we call the game a \emph{singleton game}. A \emph{state} $x$ of
the game is a vector $(x_P)_{P \in \paths}$ where $x_P$ denotes the
number of players utilizing path $P$ in state $x$, and $x_e=\sum_{P \ni
e} x_P$ is the \emph{congestion} of edge $e\in E$ in state $x$. The
latency of edge $e$ in state $x$ is given by $\lat_e(x_e)$, and the
latency of path $P\in\paths$ is
\[\lat_P(x)=\sum_{e\in P}\lat_e(x_e)\enspace. \]
The latency of a player is the latency of the path it chooses.

For brevity, for all $P\in\paths$, let $1_P$ denote the $m$-dimensional
unit vector with the one in position $P$. In state $x$ a player has an
incentive to switch from path $P$ to path $Q$ if this would strictly
decrease its latency, \ie, if \[ \lat_P(x) > \lat_Q(x+1_Q-1_P)\enspace.
\] If no player has an incentive to change its strategy, then $x$ is at
a \emph{Nash equilibrium}. It is well
known~\cite{Rosenthal:PureNash:73}, that the set of Nash equilibria
corresponds to the set of states that minimize the potential function
\[ \Phi(x) = \sum_{e\in E} \sum_{i=1}^{x_e} \lat_e(i)\enspace. \]
In the following, let $\Phi^*=\min_x \Phi(x)$ be the minimum potential.
Note that due to our definition of the latency functions $\Phi^* > 0$.
For every path $P \in \paths$ let
\[ \lat_P^+(x) =\lat_P(x+1_P)\enspace. \]
Note that for every path $Q \in \paths$
\[ \lat_P^+(x) \ge \lat_P(x+1_P-1_Q) \enspace. \]
Additionally, let
\[ \lavg(x) =\sum_{P\in\paths} \frac{x_P}{n}\lat_P(x) \]
denote the average latency of the paths in state $x$, and let
\[ \lavg^+(x) = \sum_{P\in\paths} \frac{x_P}{n}\lat_P(x+1_P)
\enspace.\]
Finally, let $\ell_{\max} = \max_{x}\max_{P\in\paths} \lat_P(x)$ denote
the maximum latency of any path. Throughout this paper, whenever we
consider a fixed state $x$ we simply drop the argument $(x)$ from
$\Phi$, $\lat_P$, $\lat_P^+$, $\lavg$, and $\lavg^+$.

\subsection{The Elasticity and the Slope of Latency Functions}

To bound the steepness of the latency functions and the effect that overshooting
may have, we consider the elasticity of the latency functions. Let $d$ denote an
upper bound on the elasticity of the latency functions, i.\,e.,
\[ d \ge \max_{e\in E} \sup_{x\in(0,n]} \left\{\frac{\lat'_e(x)\cdot
x}{\lat_e(x)}\right\}\enspace. \]
Now given a latency function with elasticity $d$, it holds that for any $x$ and
$\alpha\ge 1$, $\lat_e(\alpha\,x) \le \lat_e(x) \cdot \alpha^d$ and for $0\le
\alpha < 1$, $\lat_e(\alpha\,x) \ge \lat_e(x) \cdot \alpha^d$. As an example, the
function $a\,x^d$ has elasticity $d$.

For almost empty resources, we will also need an upper bound on the
slope of the the latency functions.  Let $\nu_e$ denote the maximum
slope on almost empty edges, i.\,e., \[ \nu_e = \max_{x\in
\{1,\ldots,d\}}\{\lat_e(x) - \lat_e(x-1)\}\enspace.\] Finally, for
$P\in\paths$, let $\nu_P = \sum_{e\in P} \nu_e$ and choose $\nu$ such
that $\nu \ge \max_{P\in\paths} \nu_P$.

  \subsection{The Imitation Protocol}

Our \IP (Protocol~\ref{alg:imitation}) proceeds in two steps. First, a
player samples another agent uniformly at random. The player then
migrates with a certain probability from its old path $P$ to the
sampled path $Q$ depending on the anticipated relative latency gain
$(\lat_P(x)-\lat_Q(x+1_Q-1_P))/\lat_P(x)$ and on the elasticity of the
latency functions. Our analysis concentrates on dynamics that result
from the protocol being executed by the players in parallel in a
round-based fashion. These dynamics generate a sequence of states
$x(0), x(1),\ldots$. The resulting dynamics converge to a state that
is stable in the sense that imitation cannot produce further progress,
\ie, $x(t+1)=x(t)$ with probability $1$. Such a state is called an
\emph{imitation-stable state}.  In other words, a state is
imitation-stable if it is $\epsilon$-Nash with $\epsilon=\nu$ with
respect to the strategy space restricted to the current support. Here,
$\epsilon$-Nash means that no agent can improve its own payoff
unilaterally by more than $\epsilon$.

\begin{algorithm}[htbp]
  \caption{\IP, repeatedly executed by all players in parallel.}
  \label{alg:imitation}
  \begin{algorithmic}
    \STATE Let $P$ denote the path of the player in state $x$.
    \STATE Sample another
    player uniformly at random. Let $Q$ denote its path.
    \IF{$\lat_P(x) > \lat_Q(x+1_Q-1_P) + \nu$}
    \STATE with probability \[ \mu_{PQ} =
    \frac{\lambda}{d}\cdot\frac{\lat_P(x) -
    \lat_Q(x+1_Q-1_P)}{\lat_P(x)}\] migrate from path $P$ to path $Q$.
    \ENDIF
  \end{algorithmic}
\end{algorithm}

As discussed in the introduction, the main difficulty in the design of
the protocol is to bound overshooting effects. To get an intuition of
this problem, consider two parallel links of which the first has the
constant latency function $\lat_1(x)=c$ and the second has the latency
function $\lat_2(x)=x^d$. Recall that the elasticity of $\lat_2$ is
$d$. Furthermore, assume that only a small number of agents $x_2$
utilizes link $2$ whereas the majority of $n-x_2$ users utilizes link
$1$. Let $b = c - x_2^d > 0$ denote the latency difference between the
two links. A simple calculation shows that using the protocol without
the damping factor $1/d$, the expected latency increase on link $2$
would be $\Theta(b\cdot d)$, overshooting the balanced state by a
factor $d$.  For this reason, we reduce the migration probability
accordingly. The constant $\lambda$ will be determined later.

Note that the arguments in the last paragraph hold for the
\emph{expected} load changes. Our protocol, however, has to take care
of probabilistic effects, i.\,e., the realized migration vector may
differ from its expectation. Typically, we can use the elasticity to
bound the impact of this effect. However, if the congestion on an edge
is very small, i.\,e., less than $d$, then the number of joining
agents is not concentrated sharply enough around its expectation. In
order to compensate for this,
we add an additional requirement that agents only migrate if the
anticipated latency gain is at least $\nu$ and use this to bound
probabilistic effects if the congestion of the edge is less than $d$.
Let us remark that we will see below (Theorem~\ref{trm:empty}) that for
a large class of singleton games it is very unlikely, that an edge will
ever have a load of $d$ or less, so the protocol will behave in the
same way with high probability for a polynomial number of rounds even
if this additional requirement is dropped.

  
  \section{Imitation Dynamics in Games with General Strategy Spaces}
\label{Sec:ExactConvergence}

In this chapter, we consider \emph{imitation dynamics} that emerge if
in each round players concurrently apply the \IP. At first, we observe
\emph{that} imitation dynamics converge to imitation stable states
since in each round the potential $\Phi(x)$ decreases in expectation.
From this result we derive a pseudopolynomial upper bound on the
convergence time to imitation-stable states.

\subsection{Pseudopolynomial Time Convergence to Imitation-Stable States} 

Consider two states $x$ and $x'$ as well as a \emph{migration vector}
$\Delta x = (\Delta x_P)_{P\in\paths}$ such that $x'=x + \Delta x$.  We
may imagine $\Delta x$ as the result of one round of the \IP although
the following lemma is independent of how $\Delta x$ is constructed.
Furthermore, we consider $\Delta x$ to be composed of a set of
migrations of agents between pairs of paths, i.\,e., $\Delta x_{PQ}$
denotes the number of players who switch from path $P$ to path $Q$, and
$\Delta x_P$ denotes the total increase or decrease of the number of
players utilizing path $P$, that is,
\[ \Delta x_P = \sum_{Q\in\paths}(x_{QP}-x_{PQ})\enspace. \]
Also, let $\Delta x_e=\sum_{P\ni e}\Delta x_P$ denote the induced
change of the number of players utilizing edge $e\in E$. In order to
prove convergence, we define the \emph{virtual potential gain}
\[ V_{PQ}(x, \Delta x) = x_{PQ}\cdot (\lat_Q(x + 1_Q - 1_P) -
\lat_P(x)) \]
which is the sum of the potential gains each player migrating from path
$P$ to path $Q$ would contribute to $\Delta \Phi$ if each of them was
the only migrating player. Note that if a player improves the latency
of his path, the potential gain is negative. The sum of all virtual
potential gains is a very rough lower bound on the true potential gain
$\Delta \Phi(x, \Delta x) = \Phi(x+\Delta x) - \Phi(x)$. In order to
compensate for the fact that players concurrently change their
strategies, consider the \emph{error term} on an edge $e \in E$:
\[ F_e(x,\Delta x) = \left\{
  \begin{array}{ll}
    \displaystyle
    \sum_{u=x_e+1}^{x_e+\Delta x_e} \lat_e(u) - \lat_e(x_e+1) & \quad\text{if
    $\Delta x_e>0$} \\
    \displaystyle
    \sum_{u=x_e+\Delta x_e+1}^{x_e} \lat_e(x_e) - \lat_e(u) & \quad\text{if
    $\Delta x_e<0$} \\
    \displaystyle
    0 & \quad\text{if $\Delta x_e=0$} 
  \end{array}
\right.
\]
Subsequently, we show that the sum of the virtual potential gains and
the error terms is indeed an upper bound on the true potential gain
$\Delta \Phi(x, \Delta x)$. A similar result is shown
in~\cite{Fischer/Voecking:Stale:TCS:08} for a continuous model.

\begin{lemma}
  \label{trm:errorterm}
  For any assignment $x$ and migration vector $\Delta x$ it holds that
  \[
  \Delta \Phi(x, \Delta x)
  \, \le \,
  \sum_{P,Q\in\paths} V_{PQ}(x, \Delta x ) + \sum_{e\in E} F_e(x, \Delta x)
  \enspace.
  \]
\end{lemma}

\begin{proof}
  We first express the virtual potential gain in terms of latencies on
  the edges. Clearly, 
  \begin{eqnarray}
    \nonumber
    \sum_{P,Q\in\paths} V_{PQ}(x,\Delta x) 
    &=&
    \sum_{P,Q\in\paths} x_{PQ}\cdot (\lat_Q(x + 1_Q - 1_P) - \lat_P(x))\\
    &\le&    
    \nonumber 
    \sum_{P,Q\in\paths} x_{PQ}\cdot \left(\sum_{e\in Q} \lat_e(x_e+1) - \sum_{e\in P} \lat_e(x_e)\right)\\
    &\le&     
    \label{eqn:vpg_edges}
    \sum_{e:\Delta x_e > 0} \Delta x_e \cdot \lat_e(x_e+1) + \sum_{e:\Delta x_e < 0} \Delta x_e \cdot \lat_e(x_e)
    \enspace.
  \end{eqnarray}
  The true potential gain, however, is
  \begin{eqnarray*}
    \Delta\Phi(x,\Delta x)
    &=&
    \sum_{e:\Delta x_e>0} \;\;\; \sum_{u=x_e+1}^{x_e+\Delta x_e} \lat_e(u) -
    \sum_{e:\Delta x_e<0} \;\;\; \sum_{u=x_e-\Delta x_e+1}^{x_e} \lat_e(u)\\
    &=&
    \sum_{e:\Delta x_e>0} \left(\Delta x_e\cdot\lat_e(x_e+1) + \sum_{u=x_e+1}^{x_e+\Delta x_e} (\lat_e(u) - \lat_e(x_e+1))\right) \\
    && +
    \sum_{e:\Delta x_e<0} \left(\Delta x_e\cdot\lat_e(x_e) + \sum_{u=x_e-\Delta x_e+1}^{x_e} (\lat_e(x_e) - \lat_e(u)) \right)
    \enspace.
  \end{eqnarray*}
  Substituting Equation~(\ref{eqn:vpg_edges}) for the left term of
  each sum and the definition of $F_e$ for the right term of each sum,
  we obtain the claim of the Lemma.
\end{proof}


In the following, we consider $\Delta x$ to be a migration vector
generated by the \IP rather than an arbitrary vector. In this case,
$\Delta x$ is a random variable and all probabilities and expectations
are taken with respect to the \IP. In order to prove that the
potential decreases in expectation, we derive a bound on the size of
the error terms. We show that the error terms reduce the virtual
potential gain by at most a factor of two, or, put another way, that
the true potential gain is at least half of the virtual potential
gain.

\begin{lemma} \label{trm:virtual_potential} 
  Let $x$ denote a state and let the random variable $\Delta x$ denote
  a migration vector generated by the \IP. Then,
  \begin{eqnarray*}
    \Ex{ \Delta \Phi(x,\Delta x) } & \le &
    \frac{1}{2} \sum_{P,Q\in\paths} \Ex{V_{PQ}(x,\Delta x)} \enspace.
  \end{eqnarray*}
\end{lemma}

\begin{proof}

  For any given round, each term in $V_{PQ}$, $P,Q\in\paths$ and $F_e$, $e\in E$
  can be associated with an agent. Fix an agent $i$ migrating from, say, $P$ to
  $Q$. Its contribution to the $V_{PQ}(x,\Delta x)$ is $\lat_Q(x+1_Q-1_P) -
  \lat_P(x)$ (this is the same for all agents moving from $P$ to $Q$). It also
  contributes to $F_e$, $e\in P\cup Q$. The size of this term depends on the
  ordering of the agents. We will consider the migrating agents in ascending
  order of the migration probability $\mu_{P_jQ_j}$, where $P_j$ and $Q_j$ denote
  the origin and destination path of agent $j$, respectively. Ties are broken
  arbitrarily.

  Fix an edge $e\in E$ and let $A^+(e)$ and $A^-(e)$ denote the set of
  agents migrating to and away from $e\in E$, respectively.  Let
  $A(e)=A^+(e)\cup A^-(e)$.  Let $\DX$ denote the contribution to
  $\Delta x_e$ of agents in $A(e)$ which occur in our ordering with
  respect to $\mu_{PQ}$ before agent $i$.

  \begin{figure}[htpb]
    \begin{center}
      \psset{linewidth=1pt,xunit=1,yunit=1}
      \begin{pspicture}(0,0)(13,7)
        \psline{->}(0,0)(0,5)\uput[-90](0,5.6){$\lat_e(x)$}
        \psline{->}(0,0)(6,0)\uput[-90](6,-0.2){$x$}
        \psline{-}(2,0.1)(2,-0.1)
        \psline{-}(4.5,0.1)(4.5,-0.1)
        \uput[0](3.25,-0.2){$\DX$}
        \uput[0](2,-0.2){$x_e$}
        \psplot[linewidth=2pt,linecolor=red,plotpoints=100]{0}{5.5}{x x 0.1 mul mul}
        \psframe[fillstyle=vlines](4.5,0.4)(5,2.5)
        \psline[linestyle=dotted]{-}(0,0.4)(12,0.4)
        \psline{-}(2,0.4)(5,0.4)
        \psline{->}(7,0)(7,5)\uput[-90](7,5.6){$\lat_{e'}(x)$}
        \psline{->}(7,0)(12,0)\uput[-90](12,-0.2){$x$}
        \psplot[linewidth=2pt,linecolor=red,plotpoints=100]{7}{12}{x 7 sub 0.3 mul 3 add}
        \psframe[fillstyle=crosshatch](9.5,0.4)(10,3.9)
      \end{pspicture}
    \end{center}
    \caption{Potential gain of an agent migrating from edge $e'$
      towards edge $e$. The hatched area is the agent's virtual
      potential gain. The shaded area on the left is this agents
      contribution to the error term, caused by the $\DX$ agents
      ranking before the agent under consideration (with respect to
      $\mu_{PQ}$).}
    \label{fig:error_terms}
  \end{figure}
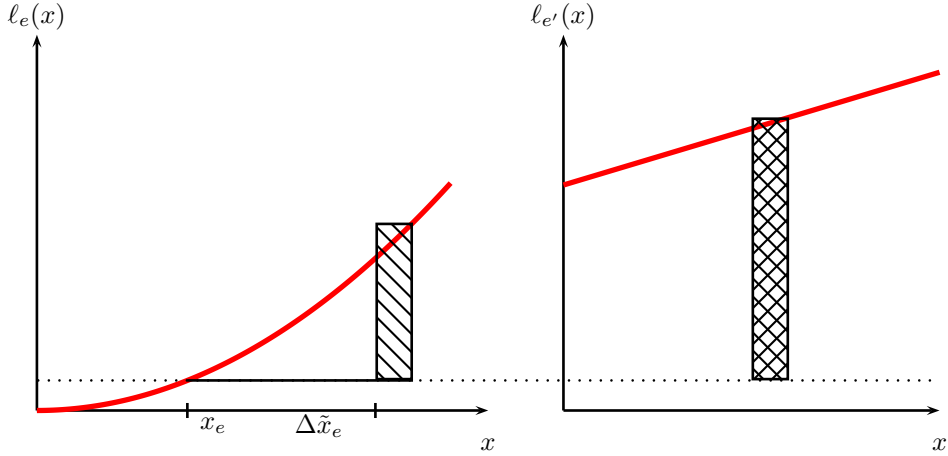
  Agent $i$'s contribution to $F_e(x,\Delta x)$ is $\Delta\tilde\lat_e(\DX)$
  where we define the error function $\Dl(\delta) = \lat_e(x_e+1+\delta) -
  \lat_e(x_e+1)$. For an illustration, see
  Figure~\ref{fig:error_terms}.  Note that there is an exception: If
  $e\in Q\cap P$, then the contribution of agent $i$ to $F_e$ is zero
  and there is nothing to show.  For brevity, let us write
  $\lat_e=\lat_e(x_e)$ and $\lat_e^+=\lat_e(x_e+1)$ as well as
  $\lat_P=\lat_P(x)$ and $\lat_Q^+=\lat_P(x_e+1_Q-1_P)$.  For $e\in
  Q\setminus P$ we show that
  \begin{eqnarray}
    \label{eqn:tilde_l_1}
    \Ex{\Delta\tilde\lat_e\left(\Delta\tilde x_e\right)}
    &\le& \frac18 \cdot (\lat_P - \lat_Q^+)\cdot \left(\frac{\lat_e^+}{\lat_Q^+}
    + \frac{\nu_e}{\nu_Q}\right) \enspace,
  \end{eqnarray}
  and for $e\in P\setminus Q$,
  \begin{eqnarray}
    \label{eqn:tilde_l_2}
    \Ex{\Delta\tilde\lat_e\left(\Delta\tilde x_e\right)}
    &\le& \frac18 \cdot (\lat_P - \lat_Q^+)\cdot \left(\frac{\lat_e}{\lat_P} +
    \frac{\nu_e}{\nu_P}\right)\enspace.
  \end{eqnarray}
  Thus, the expected sum of the error terms of an agent migrating from
  $P$ to $Q$ is at most
  \begin{eqnarray*}
    \frac{\lat_P-\lat_Q^+}{8}
    \left(\sum_{e\in P\setminus Q} \left(\frac{\lat_e}{\lat_P}+\frac{\nu_e}{\nu_P}\right)+
      \left(\sum_{e\in Q\setminus P} \frac{\lat_e^+}{\lat_Q^+}+\frac{\nu_e}{\nu_Q}\right)\right)
    &\le&
    \frac12 (\lat_P-\lat_Q^+)
    \enspace,
  \end{eqnarray*}
  \ie, half of its virtual potential gain, which proves the lemma.
  First, consider the case that $e\in Q$ where $Q$ denotes the
  destination path of agent $i$.

  For brevity, let us write $I_{PQ}=(\lat_P - \lat_Q^+)/\lat_P$ for the incentive
  to migrate from $P$ to $Q$. Again, consider the case that $e\in Q$ where $Q$
  denotes the destination path of agent $i$. Then, due to our ordering of the
  agents,
  \begin{equation}
    \label{eqn:exp_upper}
    \Ex{\Delta \tilde x_e}
    \le 
    n\cdot \frac{x_e}{n} \cdot \mu_{PQ}
    \le \frac{\lambda\cdot x_e\cdot I_{PQ}}{d}
    \enspace,
  \end{equation}
  implying
  \begin{equation}
    \label{eqn:x_lower}
    x_e 
    \ge 
    \frac{\Ex{\Delta \tilde x_e}\cdot d}{\lambda\cdot I_{PQ}}
    \enspace.
  \end{equation}
  Furthermore, due to the elasticity of $\lat_e$, and using $(1+1/x)^x \le \exp(1)$, we obtain
  \begin{eqnarray}
    \Delta\tilde\lat_e(\delta)
    \nonumber &\le& 
    \lat_e^+ \cdot\left(\frac{x_e+1+\delta}{x_e+1}\right)^d - \lat_e\\
    \nonumber &\le&
    \lat_e^+ \cdot\left(1+\frac{\delta}{x_e}\right)^d - \lat_e^+\\
    \label{eqn:delta_l_upper}
    &\le& \lat_e^+ \cdot \left(\expf{\frac{d\,\delta}{x_e}} - 1\right)
    \enspace.
  \end{eqnarray}
  Subsequently, we consider two cases.
  \begin{description}
  \item[Case 1: $\Ex{\DX} \ge \frac{1}{64}$.]  Substituting
    Inequality~(\ref{eqn:x_lower}) into Inequality~(\ref{eqn:delta_l_upper}), 
    we obtain for every $\kappa \in \RR-{\ge 0q}$
    \begin{eqnarray*}
      \Dl\left(\kappa \,\Ex{\DX}\right) &\le& 
      \lat_e^+\cdot \left(\expf{\kappa \,\lambda\,I_{PQ}} - 1\right) \enspace .
    \end{eqnarray*}
    Now, note that for every $k \in \NN$ and $\kappa \in [k,k+1]$
    \begin{eqnarray*}
      \Pr{\DX \ge \kappa \,\Ex{\DX}} & \le & \Pr{\DX \ge k \,\Ex{\DX}}
      \text{ \,\,\, and } \\
       \Dl(\kappa\,\Ex{\DX})& \le & \Dl((k+1)\,\Ex{\DX})
    \end{eqnarray*}
    hold. Applying a Chernoff bound (Fact~\ref{trm:chernoff} in the
    appendix), we obtain an upper bound for the expectation of
    $\Ex{\Delta\tilde\lat_e \left(\Delta\tilde x_e\right)}$ as
    follows.
    \begin{eqnarray*}
      \Ex{\Delta\tilde\lat_e\left(\Delta\tilde x_e\right)} 
      &\le& \sum_{k=1}^{\infty} \Pr{\DX \ge k\,\Ex{\DX}}
      \cdot \Dl((k+1)\,\Ex{\DX})\\
      &\le& \Delta\tilde\lat_e^+\left(5\,\Ex{\Delta\tilde x_e}\right) +
      \sum_{k=5}^{\infty} \Pr{\DX \ge k\,\Ex{\DX}}\cdot \Dl((k+1)\,\Ex{\DX})\\
      &\le& \lat_e^+\cdot \left(\expf{5\,\lambda\,I_{PQ}} - 1\right) +
      \sum_{k=5}^\infty \expf{-\frac14\,\Ex{\DX}\,k\,\ln k}\cdot
      \lat_e^+\cdot\,\left(\expf{(k+1)\,\lambda\,I_{PQ}} - 1\right)\\ 
      &\le& \lat_e^+\cdot \left(\expf{5\,\lambda\,I_{PQ}} - 1\right) +
      \sum_{k=5}^\infty \expf{-\frac14\,\Ex{\DX}\,k}\cdot
      \lat_e^+\cdot\,\left(\expf{2\,k\,\lambda\,I_{PQ}} - 1\right)\\ 
      &\le& \lat_e^+\cdot \left(\expf{5\,\lambda\,I_{PQ}} - 1\right) +
      \int_4^\infty \expf{-\frac14\,\Ex{\DX}\,u}\cdot
      \lat_e^+\cdot\,\left(\expf{2\,u\,\lambda\,I_{PQ}} - 1\right)\, du\\ 
      &=& \lat_e^+\cdot \left(\expf{5\lambda I_{PQ}} - 1 + 
        \expf{-\EDX}\frac{\expf{8\,\lambda\,I_{PQ}}-1 +
        \frac{8\,\lambda\,I_{PQ}}{\EDX}}{\frac{1}{4}\EDX - 2\,\lambda\,I_{PQ}}
        \right) \enspace.
    \end{eqnarray*}
    Now, due to Fact~\ref{trm:exp_lin} in the appendix (with $r=1$)
    and our assumption that $\Ex{\DX} \ge 1/64$, we obtain
    \begin{eqnarray*}
      \Ex{\Delta\tilde\lat_e\left(\Delta\tilde x_e\right)} 
      &\le&
      \lambda \cdot \lat_e^+\cdot I_{PQ}\cdot \left(5\,(\e-1) + \frac{8\,(\e-1)
      + 8\cdot 64}{\frac{1}{4\cdot 64} - 2\,\lambda}\right)\\ 
      &\le&
      c\cdot \lambda \cdot \lat_e^+\cdot \frac{\lat_P - \lat_Q^+}{\lat_P}\\
      &\le&
      c\cdot\,\lambda \cdot \lat_e^+\cdot \frac{\lat_P - \lat_Q^+}{\lat_Q^+}
    \end{eqnarray*}
    for some constant $c$. The first inequality holds if $\lambda<1/512$,
    proving Equation~(\ref{eqn:tilde_l_1}) if $\lambda$ is chosen small enough.
  \item[Case 2: $\Ex{\Delta\tilde x_e} < \frac{1}{64}$.] Again, in this case we
  can apply a Chernoff bound (Fact~\ref{trm:chernoff}) to upper bound
  $\Ex{\Delta\tilde\lat_e\left(\Delta\tilde x_e\right)}$.
  \begin{eqnarray*}
      \Ex{\Delta\tilde\lat_e\left(\Delta\tilde x_e\right)} 
      &\le& \sum_{k=1}^n \Pr{\DX = k} \cdot \Dl(k)\\
      &\le& \sum_{k=1}^n \Pr{\DX\ge\frac{k}{\Ex{\DX}}\,\Ex{\DX}}\cdot\Dl(k)\\ 
      &\le& \sum_{k=1}^n \expf{-k\,(\ln(k/\Ex{\DX})-1)} \cdot \Dl(k)\\
  \end{eqnarray*}
  
  There are two sub-cases:
    \begin{description}
    \item[Case 2a: $x_e > d$.] In order to bound the expected latency increase,
    we apply the elasticity bound on $\lat_e$: 
      \begin{eqnarray*}
        \Ex{\Dl(\DX)} &\le&  
        \sum_{k=1}^n \expf{- k\,(\ln(k/\Ex{\DX})-1)}\cdot
        \lat_e^+\cdot\left(\expf{\frac{k\,d}{x_e}}-1\right)\\
        &\le& 
        \lat_e^+\cdot 
          \sum_{k=1}^n \expf{- k\,(\ln(k) - \ln(\Ex{\DX})-1)} \cdot
          \left(\expf{\frac{k\,d}{x_e}}-1
        \right)\\
        &\le& 
        \lat_e^+\cdot 
          \sum_{k=1}^n \left(\Ex{\DX} \, (\expf{k}\, \Ex{\DX}^{k-1})\right)
          \,\expf{-k\,(\ln k)}\cdot \left(\expf{\frac{k\,d}{x_e}}-1
          \right)\\
        &\le& 
        \lat_e^+\cdot \Ex{\DX} \cdot \sum_{k=1}^n \expf{-k\,(\ln k)} \cdot 
        \left(\expf{\frac{k\,d}{x_e}}-1 \right)  
        \enspace.
      \end{eqnarray*}      
      Now, splitting up the sum, we define
      \begin{eqnarray*}
        L_1 &=& \Ex{\DX}
        \sum_{k=1}^{\left\lfloor\frac{8\,x_e}{d}\right\rfloor} \expf{-k\,(\ln
        k)}\cdot \left(\expf{\frac{k\,d}{x_e}}-1\right)\\
        &\le& \Ex{\DX} \frac{(\expf{8}-1)\,d}{8\,x_e}
        \sum_{k=1}^{\left\lfloor \frac{8\,x_e}{d} \right\rfloor}
        \expf{-k\,(\ln k)}\cdot k \\
        L_1 &\le& \frac{\expf{8}}{4} \cdot \Ex{\DX} \frac{d}{x_e}\\ 
        &\le& \frac{\expf{8}}{4} \cdot \lat_e^+ \cdot
        \lambda\,I_{PQ}\enspace,
      \end{eqnarray*}

      where the first inequality uses the observation that
      $\expf{\frac{k\,d}{x_e}} \le \expf{8}$ since $k \le \left\lfloor 8x_e/d
      \right\rfloor$, and Fact~\ref{trm:exp_lin} (with $r=8$). Additionally,
      where the third inequality uses the observation that $\sum_{k=1}^{\infty}
      \expf{-k\,(\ln k)}\cdot k \le 2$, and finally where the last
      inequality uses Inequality~(\ref{eqn:exp_upper}). 
      
      For the second part of the sum, let
      \begin{eqnarray*}
        L_2 
        &=& \Ex{\DX}
        \sum_{k=\left\lceil\frac{8\,x_e}{d}\right\rceil}^{\infty}
        \expf{-k\,(\ln k)} \cdot \left(\expf{\frac{k\,d}{x_e}}-1\right)\\ 
        &\le& \Ex{\DX}
        \sum_{k=\left\lceil\frac{8\,x_e}{d}\right\rceil}^{\infty}
        \expf{-k\,(\ln k) + \frac{k\,d}{x_e}}\\
        &=& \Ex{\DX}
        \sum_{k=\left\lceil\frac{8\,x_e}{d}\right\rceil}^{\infty}
        \expf{-k\,(\ln k - 1)} 
        \mbox{ \hspace{6ex} (since $x_e>d$)}\\
        &\le& \Ex{\DX}
        \sum_{k=\left\lceil\frac{8\,x_e}{d}\right\rceil}^{\infty}
        \expf{-\frac12\,k\,\ln k}
        \mbox{ \hspace{8ex} (since
        $k \ge \left\lceil\frac{8\,x_e}{d}\right\rceil \ge 8)$}\\
        &\le& \Ex{\DX} 
		\sum_{k=\left\lceil\frac{8\,x_e}{d}\right\rceil}^{\infty}
        \left(\frac{d}{8\,x_e}\right)^{\frac{1}{2} k} \enspace.
      \end{eqnarray*}
      Due to Fact~\ref{trm:geomSum} and since $x_e > d$              
      \begin{eqnarray*}
        L_2 &=& \Ex{\DX}
        \frac{\left(\frac{d}{8\,x_e}\right)^{\frac{8}{2}}}{1 -
        \sqrt{\frac{d}{8\,x_e}}}\\
        &\le& \Ex{\DX} \frac{d}{x_e}\\
        &\le& \lambda\,I_{PQ} \enspace.
      \end{eqnarray*}
      Reassembling the sum, we obtain
      \begin{eqnarray*}
        \Ex{\Dl(\DX)} &\le& \lat_e^+\cdot 
        \left(L_1 + L_2 \right)\\
        &\le& \lat_e^+\cdot \left(\frac{\expf{8}}{4} + 1 \right) \,
        \lambda\,I_{PQ} \enspace.
      \end{eqnarray*}
      Again, by the same arguments as at the end of Case~1 this proves
      Equation~(\ref{eqn:tilde_l_1}) if $\lambda$ is less than $
      1/(2\expf{8}+8)$.

    \item[Case 2b: $x_e \le d$.] In this case we separate the upper
      bound on $\Ex{\Dl(\DX)}$ into the section up to $d$ and above
      $d$.  For the first section we use the fact that each additional
      player on resource $e$ causes a latency increase of at most
      $\nu_e$ as long as the load is at most $d$. We define the
      contribution to the expected latency increase by the events that
      up to $d-x_e$ join resource $e$, i.\,e., afterwards the
      congestion is still at most $d$. In this case, we may use $\nu_e$
      to bound the contribution of each agent:
      \begin{eqnarray*}
        L_1 &\le& \sum_{k=1}^{d-x_e} \expf{-k\,\left(\ln\left(\frac{k}{\Ex{\DX}}\right) - 1\right)} \cdot k\,\nu_e\\
        &\le& \e\,\nu_e\,\Ex{\DX} + \nu_e\,\Ex{\DX}^2
        \sum_{k=2}^{d-x_e}\expf{-k\,(\ln(k) - 1)} \cdot k\\ 
        &\le& \e\,\nu_e\,\Ex{\DX}\cdot \left(1 + \frac{8\,\Ex{\DX}}{\expf{}}\right)\\
         &\le& 3\,\nu_e\,\Ex{\DX}\enspace,
      \end{eqnarray*}
      where the third inequality holds since
      $\sum_{k=2}^{d-x_e}\expf{-k\,(\ln(k) - 1)} \cdot k \le 8$, and where the
      last inequality holds since $\Ex{\DX} < 1/64$.

      For the contribution of the agents increasing the load on
      resource $e$ to above $d$ we use the elasticity constraint
      again. This time, we do not consider the latency increase with
      respect to $\lat_e^+(x_e)$ but with respect to $\lat_e(d)$:
      \begin{eqnarray*}
        L_2
        &=& \sum_{k=d-x_e+1}^{n} \expf{-k\cdot\left(\ln\left(\frac{k}{\Ex{\DX}}\right)-1\right)}\cdot \lat_e(d)\cdot\left(\expf{\frac{d\,(k-(d-x_e))}{d}}-1\right)\enspace.
      \end{eqnarray*}
      As in case (2a),
      \begin{eqnarray*}
        L_2
        &\le& 
        \lat_e(d)\cdot \Ex{\DX}\cdot \sum_{k=d-x_e+1}^{\infty}\expf{-k\,\ln k +
        k - (d-x_e)}\\ &=& 
        \lat_e(d)\cdot \Ex{\DX}\cdot
        \sum_{k=1}^{\infty}\expf{-(k+(d-x_e))\,\ln(k+(d-x_e)) + k}\\
        &=& 
        \lat_e(d)\cdot \Ex{\DX}\cdot \expf{-(d-x_e)}\cdot
        \sum_{k=1}^{\infty}\expf{-(k+(d-x_e))\,\ln(k+(d-x_e)) + k +
        d-x_e} \enspace.
      \end{eqnarray*}
      Consider the series in the above expression as a function of
      $u=(d-x_e)$ and denote it by $S(u)$. Note that $S(u)$ converges
      for every $u\ge 0$ and $S(u)\to 0$ as $u\to\infty$. In
      particular, $S(u)< 8$ for any $u\ge 0$, so  
      \begin{eqnarray*}
        L_2
        &\le& 
        8\,\lat_e(d)\cdot \Ex{\DX} \cdot \expf{-(d-x_e)}\\
        &\le& 
        8\,(\lat_e(x_e) + (d-x_e)\,\nu_e)\cdot \Ex{\DX} \cdot \expf{-(d-x_e)}\enspace.
      \end{eqnarray*}
      Since $(d-x_e)\cdot\expf{-(d-x_e)}<1/2$,
      \begin{eqnarray*}
        L_2 \le 4\,(\lat_e(x_e) + \,\nu_e)\cdot \Ex{\DX}\enspace.
      \end{eqnarray*}
      Altogether, 
      \begin{eqnarray*}
        \Ex{\Dl(\DX)} 
        &\le& L_1 + L_2\\
        &\le& 7\,\nu_e\,\Ex{\DX} + 4\,\lat_e(x_e)\,\Ex{\DX}\\
        &\le& 7\,\nu_e\,\Ex{\DX} +
        4\,\frac{\lambda\,x_e\,I_{PQ}}{d} \cdot \lat_e(x_e)\\
        &\le& \frac{7}{64}\,\nu\,\frac{\nu_e}{\nu_Q} +
        \frac{4\,\lambda\,x_e\,I_{PQ}}{d}\cdot\lat_e(x_e)
      \end{eqnarray*}
      where we have used Equation~(\ref{eqn:exp_upper}) for the third
      inequality, and the inequalities $\Ex{\DX}<1/64$ and $\nu \ge \nu_Q$ for
      the last step. Since $x_e\le d$ and $\lat_P-\lat_Q^+\ge\nu$,
      \begin{eqnarray*}
        \Ex{\Dl(\DX)} 
        &\le& \frac{1}{8}\,(\lat_P-\lat_Q^+)\,\frac{\nu_e}{\nu_Q} +
        \frac{4\,\lambda\,(\lat_P-\lat_Q^+)}{\lat_P}\cdot\lat_e(x_e)\\
      \end{eqnarray*}
      again proving Equation~(\ref{eqn:tilde_l_1}) if $\lambda\le 1/32$. 
    \end{description}    
  \end{description}
  Finally, the case $e\in P$ is very similar.
\end{proof}



Note that all migrating players add a negative contribution to the
virtual potential gain since they migrate only from paths with
currently higher latency to paths with lower latency. Hence, together
with Lemma~\ref{trm:virtual_potential}, we can derive the next
corollary.

\begin{corollary}
  Consider a symmetric network congestion game\/ $\Gamma$ and let $x$ and $x'$
  denote states of\/ $\Gamma$ such that $x'$ is a random state generated after
  one round of executing the \IP. Then, 
  \[ \Ex{\Phi(x')} \le \Phi(x) \] 
  with strict inequality as long as $x$ is not imitation-stable. Thus,
  $\Phi$ is a super-martingale.
\end{corollary} 

It is obvious \emph{that} the sequence of states generated by the \IP
terminates at an imitation-stable state. From
Lemma~\ref{trm:virtual_potential} we can immediately derive an upper
bound on the time to reach such a state. However, since for arbitrary
latency functions the minimum possible latency gain may be very small,
this bound can clearly be only pseudo-polynomial. To see this, consider
a state in which only one player can make an improvement. Then, the
expected time until the player moves is inverse proportional to its
latency gain.

\begin{theorem} \label{trm:convergence}
  Consider a symmetric network congestion game in which all players
  use the \IP. Let $x$ denote the initial state of the dynamics. Then the
  dynamics converge to an imitation-stable state in expected time
  \[ \Oh{ \frac{d\,n\,\lat_{\max}\,\Phi(x)}{\nu^2}} \enspace. \]
\end{theorem}
\begin{proof}
  By definition of the $\IP$, the expected virtual potential gain in any
  state $x'$ which is not yet imitation-stable is at least
  \[ \Ex{\sum_{P,Q\in\paths} V_{PQ}(x',\Delta x')} \le -\nu \cdot
  \frac{\lambda}{d\,n} \cdot \frac{\nu}{\lat_{\max}}\enspace. \]
  Hence, also the expected potential gain $\Ex{\Delta \Phi(x')}$ in
  every intermediate state $x'$ of the dynamics is bounded from above
  by at least half of the above value. From this, it follows, that the
  expected time until the potential drops from at most $\Phi(x)$ to
  the minimum potential $\Phi^*$ is at most
  \[ \frac{d\,n\,\lat_{\max}(\Phi(x) - \Phi^*)}{\lambda\,\nu^2} \enspace. \] 
  Formally, this is a consequence of Lemma~\ref{trm:martingalelike}
  which can be found in the Appendix.
\end{proof}

It is obvious that this result cannot be significantly improved since
we can easily construct an instance and a state such that the only
possible improvement that can be made is $\nu$. Hence, already a
single step takes pseudopolynomially long. In case of polynomial
latency functions Theorem~\ref{trm:convergence} reads as follows.

\begin{corollary}
  Consider a symmetric network congestion game with
  polynomial latency functions with maximum degree $d$ and minimum and
  maximum coefficients $a_{\min}$ and $a_{\max}$, respectively.  Let
  $k=\max_{P\in\paths}{|P|}$. Then the dynamics converges to an
  imitation-stable state in expected time
  \[ \Oh{ d^2\,k^2\,n^{2d+2}\cdot
    \left(\frac{a_{\max}}{a_{\min}}\right)^2} \enspace. \]
\end{corollary}

Let us remark that all proofs in this section do not rely on the
assumption that the underlying congestion game is symmetric. In fact,
the lemma also holds for asymmetric congestion games in which each
player samples only among players that have the same strategy space.

  \subsection{Sequential Imitation Dynamics and a Lower Bound}\label{sequential}

In the previous section, we proved that players applying the \IP reach
an imitation-stable state after a pseudopolynomial number of
rounds. Recall that in this case each player decreases its latency by
at least $\nu$ if it were the only player to change its strategy. In
this section, we consider sequential imitation dynamics such that in
each round a single player is permitted to imitate someone
else. Furthermore, we assume that each player changes its path
regardless of the anticipated latency gain. Now, it is obvious that
sequential imitation dynamics converge towards imitation-stable states
as the potential $\Phi$ strictly decreases after every strategy
change. Hence, we focus on the convergence time of such dynamics.

For such sequential imitation dynamics we prove an exponential lower
bound on the number of rounds to reach an imitation-stable state. To
be precise, we present a family of symmetric network congestion games
with corresponding initial states such that every sequence of
imitation leading to an imitation-stable state is exponentially
long. To some extent, this results complements
Theorem~\ref{trm:convergence} as it presents an exponential lower
bound in a slightly different model. However, in this lower bound
$\nu$ is arbitrary large and almost every state is imitation-state
with respect to the \IP.

\begin{theorem}\label{Thm:SequentialImitation}
  For every $n \in \NN$, there exists a symmetric network congestion game with
  $n$ players, initial state $\Sinit$, polynomial bounded network size, and
  linear latency functions such that \emph{every} sequential imitation dynamics
  that start in $\Sinit$ is exponentially long.
\end{theorem}

Subsequently, we do not give a complete proof of the theorem but we
discuss how to adapt a series of constructions as presented
in~\cite{Ackermann/etal:Combinatorial:06} which shows that there
exists a family of symmetric network congestion games with the same
properties as stated in the above theorem such that \emph{every best
  response dynamics} starting in $\Sinit$ is exponentially long. To be
precise, they prove that in every intermediate state of the best
response dynamics \emph{exactly} one player can improve its
latency. Recall that in best response dynamics players know the entire
strategy space and that in each round one player is permitted the
switch to the best available path.

In the following, we summarize the constructions presented
in~\cite{Ackermann/etal:Combinatorial:06}. At first, a \PLS-reduction
from the local search variant of {\sf MaxCut} to threshold games is
presented. In a threshold game, each player either allocates a single
resource on its own or shares a bunch of resources with other players.
Hence, in a threshold game each player chooses between two strategies
only. The precise definition of these games is given below. Then, a
\PLS-reduction from threshold games to asymmetric network congestion
games is presented. Finally, the authors
of~\cite{Ackermann/etal:Combinatorial:06} show how to transform an
asymmetric network congestion game into a symmetric one such that the
desired properties of best response dynamics are preserved. All
\PLS-reductions are embedding, and there exists a family of instances
of {\sf MaxCut} with corresponding initial configurations such that in
every intermediate configuration generated by a local search algorithm
exactly one node can be moved to the other side of the cut. Therefore,
there exists a family of symmetric network congestion games with the
properties as stated above.

A naive approach to prove a lower bound on the convergence time of
imitation dynamics in symmetric network congestion games is as
follows. Building upon the lower bound of the convergence time of best
responses dynamics, a player for every path is added to the game. Then
the latency functions are adopted accordingly. However, in this case
we would introduce an exponential number of additional players. In
threshold games, however, the players' strategy spaces have size two
only. Hence, we could apply this approach to threshold games. In the
following, we present the details of this approach. It is then not
difficult to verify that the \PLS-reductions mentioned above can be
reworked in order to prove
Theorem~\ref{Thm:SequentialImitation}. However, note that this does
not imply that computing a imitation-stable state is \PLS-complete
since one can always assign all players to the same strategy which
obviously is an imitation-stable state.

\emph{Threshold games} are a special class of congestion games in
which the set of resources $\calR$ can be divided into two disjoint
sets $\Rin$ and $\Rout$. The set $\Rout$ contains exactly one
resource $r_i$ for every player $i \in \calN$. This resource has a
fixed latency $T_i$ called the \emph{threshold} of player $i$. Each
player $i$ has only two strategies, namely a strategy $\Sout_i =
\{r_i\}$ with $r_i \in \Rout$, and a strategy $\Sin_i \subseteq
\Rin$. The preferences of player $i$ can be described in a simple and
intuitive way: Player $i$ prefers strategy $\Sin_i$ to strategy
$\Sout_i$ if the latency of $\Sin_i$ is smaller than the threshold
$T_i$. \emph{Quadratic threshold games} are a subclass of threshold
games in which the set $\Rin$ contains exactly one resource $r_{ij}$
for every unordered pair of players $\{i,j\} \subseteq
\calN$. Additionally, for every player $i \in \calN$ of a quadratic
threshold game, $\Sin_i = \{r_{ij} \mid j\in \calN, j\neq
i\}$. Moreover, for every resource $r_{ij} \in \Rin$:
$\lat_{r_{ij}}(x) = a_{i,j} \cdot x$ with $a_{ij} \in \NN$, and for
every resource $r_i$: $\lat_{r_i}(x) = 1/2 \, \sum_{j \neq i} a_{ij}
\cdot x$ to $r_i$.

Let $\Gamma$ be a quadratic treshold game that has an initial state
$\Sinit$, such that every best response dynamics which starts $\Sinit$
is exponentially long, and every intermediate state has a unique
player which can improve its latency. Suppose now that we replace
every player $i$ in $\Gamma$ by three players $i_1, i_2$ and $i_3$
which all have the same strategy spaces as player $i$
has. Additionally, suppose that we choose new latency functions
$\lat'$ for every resource $r_i$ as follows: $\lat'_{r_i}(x) = 1/2
\sum_{j \neq i} a_{ij} \cdot x + 3/2 \sum_{j \neq i} a_{ij}$. Hence,
we add an additional offset of $3/2 \sum_{j \neq i} a_{ij}$.

Suppose now that we assign every player $i_1$ to $\Sout_i$, and every
player $i_2 $ to $\Sin_i$. For every possible strategy that the $i_3$
players can use, their latency increases by $2 \sum_{j \neq i}
a_{ij}$, compared to the equivalent state in the original game, in
which every player $i$ chooses the same strategy as player $i_3$ does.
Hence, if if we assign every player $i_3$ to the strategy chosen by
player $i$ in $\Sinit$ and if the players $i_1$ and $i_2$ were not
permitted to change their strategies, then we would obtain the desired
lower bound on the convergence time of imitation dynamics in threshold
games. However, since also $i_1$ and $i_2$ are permitted to imitate,
it remains to show that whenever player $i_3$ has changed its
strategy, then both $i_1$ and $i_2$ do not want to change their
strategies anymore. 

First, suppose that player $i_3$ switches from the strategy of player
$i_2$ to the strategy of player $i_1$. Obviously, player $i_1$ does
not want to change its strategy as otherwise $i_3$ would not have
imitated $i_1$. Suppose now that $i_2$, whose strategy is dropped by
$i_3$, also wants to imitate $i_1$. In this case all three players
would allocate $\Sout_i$, and hence have latency $3 \, \sum_{r \in
  j\neq i} a_{ij}$. However, if player $i_2$ would stay with strategy
$\Sin$ then its latency is upper bounded by $2 \, \sum_{r \in \Sin_i}
a_{ij}$.  Hence, players $i_1, i_2, i_3$ will never select $\Sout$ at
the same time.

Second, suppose that player $i_3$ switches from the strategy of player
$i_1$ to the strategy of player $i_2$. Now, player $i_2$ does not want
to change its strategy as otherwise $i_3$ would not have imitated
$i_2$. Suppose now that $i_1$, whose strategy is dropped by $i_3$,
also wants to imitate $i_3$.  In this case, the latency would
increase to at least $3 \, \sum_{r \in j\neq i} a_{ij}$, whereas player
$i_1$ would have latency $2 \, \sum_{r \in j\neq i} a_{ij}$ if it
would stay with strategy $\Sout$. Hence, players $i_1, i_2, i_3$ will
never select $\Sin$ at the same time.

By applying the argument that all three players never allocate the
same strategy at the same point in time we can conclude our
claim and Theorem~\ref{Thm:SequentialImitation} follows.


  \section{Fast Convergence to Approximate Equilibria}
\label{Section:ApproxStable}

Theorem~\ref{trm:convergence} guarantees convergence of concurrent
imitation dynamics generated by the \IP to an imitation-stable state in
the long run. However, it does not give a reasonable bound on the time
due to the small progress that can be made. Hence, as our main result,
we present bounds on the time to reach an 
\emph{approximate equilibrium}. Here we relax the definition of an
imitation-stable state in two aspects: We allow only a small minority
of agents to deviate by more than a small amount from the average
latency. Our notion of an approximate equilibrium is similar to the
notion used in \cite{Blum/etal:NoRegret:06,
Fischer/etal:FastConvergence:06, Fotakis/etal:CGConvergence:08}.  It is
motivated by the following observation. When sampling other players
each player gets to know its latency if it would adopt that players'
strategy. Hence to some extend each player can compute the average
latency $\lavg^+$ and determine if its own latency is above or below
that average.

\begin{definition}[($\delta$,$\epsilon$,$\nu$)-equilibrium]
  Given a state $x$, let the set of \emph{expensive} paths be
  $\paths^+_{\epsilon,\nu} = \{P\in\paths: \lat_P(x) > (1+\epsilon)\,\lavg^+ +
  \nu\}$ and let the set of \emph{cheap} paths be $\paths^-_{\epsilon,\nu} =
  \{P\in\paths: \lat_P(x) < (1-\epsilon)\,\lavg - \nu\}$.
  Let $\paths_{\epsilon,\nu} =
  \paths^+_{\epsilon,\nu}\cup\paths^-_{\epsilon,\nu}$. A configuration $x$ is at
  a \emph{($\delta$,$\epsilon$,$\nu$)-equilibrium} iff it holds that
  $\sum_{P\in\paths^{\epsilon,\nu}} x_P \le \delta\cdot n$.
\end{definition}

Intuitively, a state at ($\delta$,$\epsilon$,$\nu$)-equilibrium is a state in
which almost all agents are almost satisfied when comparing their own
situation with the situation of other agents. 
One may hope that it is possible to reach a state in which \emph{all}
agents are almost satisfied quickly . This would be a relaxation of
the concept of Nash equilibrium. We will argue below, however, that
there is no rapid convergence to such states.

\begin{theorem}
  \label{trm:convergence_bicriteria}
  For an arbitrary initial assignment $x_0$, let $\tau$ denote the
  first round in which the \IP reaches a
  ($\delta$,$\epsilon$,$\nu$)-equilibrium. Then,
  \[
  \Ex{\tau} =
  \Oh{\frac{d}{\epsilon^2\,\delta}\cdot\log\left(\frac{\Phi(x_0)}{\Phi^*}\right)}
  \enspace.
  \]
\end{theorem}


\begin{proof}
  We consider a state $x(t)$ that is not at a
  ($\delta$,$\epsilon$,$\nu$)-equilibrium and derive a lower bound on
  the expected potential gain. There are two cases. Either at least
  half of the agents utilizing paths in $\paths_{\epsilon,\nu}$
  utilize paths in $\paths^+_{\epsilon,\nu}$ or at least half of them
  utilize paths in $\paths^-_{\epsilon,\nu}$.

  \begin{description}
  \item[Case 1:] Many agents use expensive paths, i.\,e., $\sum_{P\in
      \paths^+_{\epsilon,\nu}} x_P \ge \delta\,n/2$.  Let us define
    the volume $T$ and the average ex-post latency $C$ of potential
    destination paths, i.\,e., paths with ex-post latency at most
    $(1+\epsilon)\lavg^+$, by
    \[ 
    T = \sum_{Q:\lat^+_Q \le (1+\epsilon)\lavg^+} \frac{x_Q}{n}
    \quad\text{and}\quad
    C = \frac{1}{T} \sum_{Q:\lat^+_Q \le (1+\epsilon)\lavg^+} \frac{x_Q}{n} \lat^+_Q
    \enspace.
    \]
    Clearly,
    \[
    \lavg^+
    = \sum_{P}\frac{x_P}{n}\lat_P^+
    \ge T\cdot C +
    (1-T)\cdot(1+\epsilon)\,\lavg^+\enspace,
    \]
    and solving for $T$ yields
    \begin{equation}
      \label{eqn:c_t_tradeoff1}
      T \ge \frac{\epsilon\,\lavg^+}{(1+\epsilon)\,\lavg^+ - C} 
      \enspace.
    \end{equation}

    We now give a lower bound on the expected virtual potential gain
    given that the current state is not at a
    ($\delta$,$\epsilon$,$\nu$)-equilibrium. We consider only the
    contribution of agents utilizing paths in $\paths^+_{\epsilon,\nu}$
    and sampling paths with ex-post latency below
    $(1+\epsilon)\,\lavg^+$. Then,
    \begin{eqnarray*}
      \Ex{\sum_{P,Q} V_{PQ}}
      &\le& -\frac{\lambda}{d}\sum_{P\in \paths^+_{\epsilon,\nu}} x_P 
      \sum_{Q:\lat^+\le (1+\epsilon)\lavg^+} \frac{x_Q}{n}\cdot\frac{\lat_P-\lat_Q(x+1_Q-1_P)}{\lat_P}(\lat_P-\lat_Q(x+1_Q-1_P))\\
      &=& -\frac{\lambda}{d}\sum_{P\in \paths^+_{\epsilon,\nu}} x_P \lat_P
      \sum_{Q:\lat^+\le (1+\epsilon)\lavg^+} \frac{x_Q}{n}\cdot\left(\frac{\lat_P-\lat_Q^+}{\lat_P}\right)^2
      \enspace.
    \end{eqnarray*}
    Using Jensen's inequality (Fact~\ref{trm:jensen})
    and substituting $\lat_P \ge \lavg^+$ yields
    \begin{eqnarray*}
      \Ex{\sum_{P,Q} V_{PQ}}
      &\le& -\frac{\lambda}{d}\lavg^+\sum_{P\in \paths^+_{\epsilon,\nu}} x_P
      \left(\sum_{Q:\lat^+\le (1+\epsilon)\lavg^+} \frac{x_Q}{n}\cdot \frac{\lat_P-\lat_Q^+}{\lat_P}\right)^2 \cdot \frac{1}{\sum_{Q:\lat_Q^+\le(1+\epsilon)\lavg^+}\frac{x_Q}{n}}\enspace.
    \end{eqnarray*}
    Now we substitute $\lat_P\ge (1+\epsilon)\,\lavg^+$ and use the fact that the squared expression is monotone in $\lat_P$. Furthermore, we substitute the definition of $T$ and $C$ to obtain
    \begin{eqnarray*}
      \Ex{\sum_{P,Q} V_{PQ}}
      &\le& -\frac{\lambda}{d}\lavg^+\sum_{P\in \paths^+_{\epsilon,\nu}} x_P
      \left(\frac{T\,(1+\epsilon)\lavg^+-\sum_{Q:\lat^+\le (1+\epsilon)\lavg^+} \frac{x_Q\,\lat_Q^+}{n}}{(1+\epsilon)\lavg^+}\right)^2 \cdot \frac{1}{T}\\
      &\le& -\frac{\lambda}{d}\lavg^+\sum_{P\in \paths^+_{\epsilon,\nu}} x_P
      \left(\frac{T\,(1+\epsilon)\lavg^+-T\,C}{(1+\epsilon)\lavg^+}\right)^2 \cdot \frac{1}{T}\\
      &=& -\frac{\lambda}{d}\lavg^+\cdot
      \left(\frac{(1+\epsilon)\lavg^+-C}{(1+\epsilon)\lavg^+}\right)^2 \cdot T\cdot \sum_{P\in \paths^+_{\epsilon,\nu}} x_P
      \enspace.
    \end{eqnarray*}
    We can now use the tradeoff shown in Equation~(\ref{eqn:c_t_tradeoff1}),
    $C\le \lavg^+$, and $\sum_{P\in \paths^+_{\epsilon,\nu}} x_P >
    \delta\,n/2$ to obtain
    \begin{eqnarray*}
      \Ex{\sum_{P,Q} V_{PQ}}
      &\le& -\frac{\lambda}{d}\cdot\lavg^+\cdot
      \frac{(1+\epsilon)\lavg^+-C}{((1+\epsilon)\lavg^+)^2}\cdot \epsilon\,\lavg^+\cdot \sum_{P\in \paths^+_{\epsilon,\nu}} x_P\\
      &\le& -\frac{\lambda}{d}\cdot\epsilon \cdot \frac{\epsilon\,\lavg^+}{(1+\epsilon)^2} \cdot \frac{\delta\,n}{2}\\
      &\le& -\Omega\left(\frac{\epsilon^2\cdot \delta}{d} \cdot n\,\lavg^+\right)
      \enspace.
    \end{eqnarray*}
    Since $n\lavg^+ \ge \Phi$, we have by Lemma ~\ref{trm:virtual_potential}
    \[
    \Ex{\Phi(x(t+1))} 
    \le \Phi(x(t)) - \frac{1}{2}\Ex{\sum_{P,Q} V_{PQ}}
    \le \Phi(x(t)) \left(1 - \Omega\left(\frac{\epsilon^2\cdot \delta}{d}\right)\right)
    \enspace.
    \]

  \item[Case 2:] Many agents use cheap paths, i.\,e., $\sum_{P\in
      \paths^-_{\epsilon,\nu}} x_P \ge \delta\,n/2$. This time, we
    define the volume $T$ and average latency $C$ of paths which are
    potential origins of agents migrating towards
    $\paths^-_{\epsilon,\nu}$.
    \[ 
    T = \sum_{Q:\lat_Q \ge (1-\epsilon)\lavg} \frac{x_Q}{n}
    \quad\text{and}\quad
    C = \frac{1}{T} \sum_{Q:\lat_Q \ge (1-\epsilon)\lavg} \frac{x_Q}{n} \lat_Q
    \enspace.
    \]
    This time,
    \[ \lavg \le T\cdot C + (1-T)\cdot(1-\epsilon)\,\lavg \]
    implying
    \begin{equation}
      \label{eqn:c_t_tradeoff2}
      T
      \ge \frac{\epsilon\,\lavg}{C - (1-\epsilon)\,\lavg}
      \enspace.
    \end{equation}
    Similarly as in Case 1 we now give a lower bound on the contribution to the virtual potential gain caused by agents with latency at least $(1-\epsilon)\lavg$ sampling agents in  $\paths^-_{\epsilon,\nu}$.
    \begin{eqnarray*}
      \Ex{\sum_{P,Q} V_{PQ}}
      &\le& -\frac{\lambda}{d}\sum_{Q:\lat_Q\ge(1-\epsilon)\lavg}x_Q\,\lat_Q
      \sum_{P\in\paths^-_{\epsilon,\nu}}\frac{x_P}{n}\cdot\left(\frac{\lat_Q-\lat_P^+}{\lat_Q}\right)^2
      \enspace.
    \end{eqnarray*}
    we rearrange the sum, apply Jensen's inequality (Fact~\ref{trm:jensen}) to
    obtain
    \begin{eqnarray*}
      \Ex{\sum_{P,Q} V_{PQ}}
      &\le& -\frac{\lambda}{d}\sum_{P\in\paths^-_{\epsilon,\nu}}x_P
      \sum_{Q:\lat_Q\ge(1-\epsilon)\lavg}\frac{x_Q\,\lat_Q}{n}\cdot\left(\frac{\lat_Q-\lat_P^+}{\lat_Q}\right)^2\\
      &\le& -\frac{\lambda}{d}\sum_{P\in\paths^-_{\epsilon,\nu}}x_P
      \left(\sum_{Q:\lat_Q\ge(1-\epsilon)\lavg}\frac{x_Q\,\lat_Q}{n}\cdot\frac{\lat_Q-\lat_P^+}{\lat_Q}\right)^2\cdot\frac{1}{\sum_{Q:\lat_Q\ge(1-\epsilon)\lavg}\frac{x_Q\,\lat_Q}{n}}\\
      &=&  -\frac{\lambda}{d}\sum_{P\in\paths^-_{\epsilon,\nu}}x_P
      \left(\sum_{Q:\lat_Q\ge(1-\epsilon)\lavg}\frac{x_Q}{n}\cdot(\lat_Q-\lat_P^+)\right)^2\cdot\frac{1}{C\,T}\\
      &=&  -\frac{\lambda}{d}\sum_{P\in\paths^-_{\epsilon,\nu}}x_P \left(T\cdot(C-\lat_P^+)\right)^2\cdot\frac{1}{C\,T}\\
      &\le&  -\frac{\lambda}{d}\left(T\cdot(C-(1-\epsilon)\,\lavg)\right)^2\cdot\frac{1}{C\,T}\cdot\sum_{P\in\paths^-_{\epsilon,\nu}}x_P
      \enspace.
    \end{eqnarray*}
    Finally, using Equation~(\ref{eqn:c_t_tradeoff2}) and $C\,T\le\lavg$,
    \begin{eqnarray*}
      \Ex{\sum_{P,Q} V_{PQ}}
      &\le&  -\frac{\lambda}{d} \left(\epsilon\,\lavg\right)^2\cdot\frac{1}{C\,T}\cdot \sum_{P\in\paths^-_{\epsilon,\nu}}x_P\\
      &\le&  -\frac{\lambda\,\epsilon^2\,\lavg}{d}\delta n\,\\
      &\le&  -\Omega\left(\frac{\delta\,\epsilon^2\,\Phi}{d}\right)
      \enspace.
    \end{eqnarray*}
  \end{description}
In both cases, the potential decreases by at least a factor of
$(1-\Omega(\epsilon^2\,\delta/d))$ in expectation, which, by
Lemma~\ref{trm:martingalelike}, implies that the expected time to
reach a state with $\Phi(x(t))\le \Phi^*$ is at most the time stated
in the theorem.
\end{proof}



From Theorem~\ref{trm:convergence_bicriteria} we can immediately
derive the next corollary.

\begin{corollary}
  Consider a symmetric network congestion game with polynomial latency
  functions of maximum degree $d$ and with minimum and maximum
  coefficients $a_{\max}$ and $a_{\min}$, respectively. If all players use the
  \IP, then the expected convergence time of imitation dynamics to an
  ($\delta$,$\epsilon$,$\nu$)-equilibrium is upper bounded by 
  \[
  \Oh{\frac{d^2}{\epsilon^2\,\delta} \cdot \log\left(n\,m\frac{a_{\max}}{a_{\min}}\right)} 
  \enspace. 
  \]
\end{corollary}

Let us remark, that ($\delta$,$\epsilon$,$\nu$)-equilibria are
transient, i.\,e., they can be left again once they are reached, for
example, if the average latency decreases or if agents migrate towards
low-latency paths. However, our proofs actually do not only bound the
time until a ($\delta$,$\epsilon$,$\nu$)-equilibrium is reached for
the first time, but rather the expected total number of rounds in
which the system is not at a ($\delta$,$\epsilon$,$\nu$)-equilibrium.

Note that in the definition of ($\delta$,$\epsilon$,$\nu$)-equilibria
we require the majority of agents to deviate by no more than a small
amount from $\lavg^+$. This is because the expected latency of a path
sampled by an agent is $\lavg$, but the latency of the destination
path becomes larger if the agent migrates. We use $\lavg^+$ as an
upper bound in our proof, although we could use a slightly smaller
quantity in cases where the origin $Q$ and the destination $P$
intersect, namely $\lat_P(x+1_P-1_Q)$.  Using an average over $P$ and
$Q$ of this quantity rather than $\lavg^+$ would result in a slightly
stronger definition of
($\delta$,$\epsilon$,$\nu$)-equilibria. However, we go with the
definition as presented above for the sake of clarity of presentation.

\medskip

Let us conclude this section by showing that there are fundamental
limitations to fast convergence. One could hope to show fast
convergence towards a state in which \emph{all} agents are
approximately satisfied, i.\,e., $\delta=0$. However, any protocol
that proceeds by sampling either a strategy or an agent and then
possibly migrates, takes at least expected time $\Omega(n)$ to reach a
state in which all agents sustain a latency that is within a constant
factor of $\lavg^+$. To see this, consider an instance with $n=2\,m$
agents and identical linear latency functions. Now, let $x_1=3$,
$x_2=1$ and $x_i=2$ for $3\le i \le n$. Then, the probability that one
of the players currently using resource $1$ samples resource $2$ is at
most $\Oh{1/m}=\Oh{1/n}$. Since this is the only possible improvement
step, this yields the desired bound.


  \section{Imitation Dynamics in Singleton Games}
\label{sec:singleton}

In this section, we improve on our previous results and consider
imitation dynamics in the special case of singleton congestion
games. A major drawback of the \IP is that players who rely on this
protocol cannot explore the complete set of edge if the dynamics
start in a state in which some edges are unused.  Even worse, the
event that an edge becomes unused in later states, although it has
been used in the initial state, is not impossible. It is clear,
however, that when starting from a random initial distribution of
players among the edges, the probability of emptying an edge becomes
increasingly unlikely \emph{as the number of players increases}.

Subsequently, we formalize this statement in the following
sense. Consider a family of singleton congestion games over the
\emph{same} set of edges with latency functions without offsets. Then,
the probability that an edge becomes unused is exponentially small in
the number of players. To this end, consider a vector of continuous
latency functions $\L=(\lat_e)_{i\in[m]}$ with
$\lat_e:[0,1]\to\mathbb{R}_{\ge 0}$.  To use these functions for games
with a finite number of players, we have to normalize them
appropriately. For any such function $\lat \in \L$, let $\lat^n$ with
$\lat^n(x)=\lat(x/n)$ denotes the respective scaled function. We may
think of this as having $n$ agents with weight $1/n$ each. Note that
this transformation leaves the elasticity unchanged, whereas the step
size $\nu$ decreases as $n$ increases. For a vector of latency
functions $\L=(\lat_e)_{i\in[m]}$, let $\L^n=(\lat_e^n)_{i\in[m]}$.

\begin{theorem}
  \label{trm:empty}
  Fix a vector of latency functions $\L$ with $\lat_e(0)=0$ for all
  $i\in[m]$. For the singleton congestion game over $\L^n$ with $n$
  players, the probability that the \IP with random initialization
  generates a state with $x_e=0$ for some $i\in[m]$ within $\poly(n)$
  rounds is bounded by $2^{-\Omega(n)}$.
\end{theorem}

\begin{proof}
  Let $d$ denote an upper bound on the elasticity of the functions in
  $\L$, and let $\opt_{\L}=\min_{y} \{\lavg(y)\}$ where the minimum is
  taken over all $y\in\{y'\in\mathbb{R}_{\ge 0}^{m}\mid \sum_e
  y'_e=1\}$. In other words, $\opt_{\L}$ corresponds to the minimum
  average latency achievable in a fractional solution. For any
  $e\in[m]$, by continuity and monotonicity, there exists an $y_e>0$
  such that $\lat_e(y_e) < \opt_{\L}/4^d$ and $y_e < 1/m$.
  
Consider the congestion game with $n$ players and fix an arbitrary
edge $e\in[m]$. In the following, we upper bound the probability that
the congestion on edge $e$ falls below $n\,y_e/2$. First, consider the
random initialization in which each resource receives an expected
number of $n/m$ agents.  The probability that $x_e < n\,y_e/2 \le
n/(2\,m)$ is at most $2^{-\Omega(n\,y_e)}$. Now, consider any assignment $x$
with $x_j>n\,y_j/2$ for all $e\in[m]$. There are two cases.

\begin{description}
  \item[Case 1: $x_e > y_e\, n$.] Since in expectation, our policy
  removes at most a $\lambda/d$ fraction of the agents from edge $e$, the
  expected load in the subsequent round is at least $(1-\lambda/d)\,x_e$.
  Since for sufficiently small $\lambda$ it holds that $1-\lambda/d \ge
  3/4$, we can apply a Chernoff bound (Fact~\ref{trm:chernoff}) in order to
  obtain an upper bound of $2^{-\Omega(x_e)}$ for the probability that the congestion on $e$ decreases
  to below $x_e/2 \ge y_e\,n/2$.
\item[Case 2: $ y_e\, n / 2 < x_e \le y_e\, n$.] Hence,
  $\lat_e^n(x_e)\le \opt_\L/4^d$. In the following, let $n^-$ denote
  the number of agents on edges $r$ with
  $\lat_r^n(x_r+1)<\lat_e^n(x_e)$, and let $n^+$ denote the number of
  players utilizing edges with latency above $\opt_\L$. There are two
  subcases:
    \begin{description}
    \item[Case 2a: $n^-=0$.] Then, the probability that an agent leaves
      edge $e$ is $0$.
    \item[Case 2b: $n^- \ge 1$.] We first show that $n^+\ge
      4\,\max\{n^-,x_e\}$. For the sake of contradiction, assume that
      $n^+<4\,n^-$.  Now, consider an assignment where all of these
      players are shifted to edges $r$ with latency
      $\lat_r^n(x_r)<\lat_e^n(x_e)\le \opt_\L/4^d$, where edge $r$
      receives $n^+ \cdot x_r/n^-$ (fractional) players. In this
      assignment, the congestion on all edges is increased by no more
      than a factor of $n^+/n^- < 4$. Hence, due to the limited
      elasticity, this increases the latency by strictly less than a
      factor of $4^d$. Then, all edges have a latency of less than
      $\opt_\L/4\cdot 4=\opt_L$ and some have latency strictly less
      than $\opt_L$, a contradiction. The same argument also holds if
      we consider only resource $e$ rather than all resources $r$
      considered above. Hence, also $n^+\ge 4\,x_e$.
      
      Now, consider the number of players leaving edge $e$. Clearly,
      \[ 
      \Ex{\Delta X_e^-} 
      \le x_e\cdot \frac{\lambda}{d}\sum_{r:\lat^n_r(x_r+1)<\lat^n_e(x_e)}\frac{x_r}{n}
      = x_e\cdot \frac{\lambda\,n^-}{d\,n}
      \enspace.
      \]
      All players with current latency at least $\opt_{\L}$ can
      migrate to resource $e$ since the anticipated latency gain is
      larger than $\nu$.  Hence, the number of players migrating
      towards $e$, is at least
      \begin{eqnarray*}
        \Ex{\Delta X_e^+} 
        & \ge &
        \sum_{r:\lat^n_r(x_r)\ge \opt_\L} x_r \cdot
        \frac{\lambda\,x_e\cdot(\lat^n_r(x_r)-\lat^n_e(x_e+1))}{n\,d\,\lat^n_r(x_r)} \\
        & \ge &
        \frac{\lambda \, x_e}{n \, d} \cdot \sum_{r:\lat^n_r(x_r) \ge
        \opt_\L} x_r\cdot\frac{\lat^n_r(x_r)-2^d \cdot \lat^n_e(x_e)}{\lat^n_r(x_r)}\\
        & \ge & 
        \frac{\lambda \, x_e}{n \, d} \cdot (1-\frac{1}{2^d}) \cdot n^+ \\
        & \ge & 2 \cdot x_e\cdot \frac{\lambda}{d\,n}\max\{n^-,x_e\}
        \enspace.
      \end{eqnarray*}
      The third inequality holds since $\lat^n_r\ge\opt_{\L}$ and
      $\lat^n_e\le\opt_\L/4^d$ and the last inequality holds since
      $d\ge 1$. For any $T\ge 0$ it holds that
      \begin{eqnarray*}
        \Pr{\Delta X_e \ge 0}
        &\ge&
        \Pr{(\Delta X_e^+ \ge T) \wedge (\Delta X_e^- \le T)}\\
        &\ge& \left(1-\Pr{\Delta X_e^+ < T}\right) \cdot \left(1-\Pr{\Delta X_e^- > T}\right)
        \enspace.
      \end{eqnarray*}
      Due to our lower bounds on $\Ex{\Delta X_e^+}$ and $\Ex{\Delta
      X_e^-}$ we can apply a Chernoff bound (Fact~\ref{trm:chernoff}) on
      these probabilities. We set
      $T=1.5\,\lambda\,\max\{x_e,n^-\}\,x_e/(d\,n)$ which is an upper
      bound on $\Ex{\Delta X_e^-}$ and a lower bound on $\Ex{\Delta
      X_e^+}$, so
      \begin{eqnarray*}
        \Pr{\Delta X_e^+ < T}
        &\le& 2^{-\Omega(T)}
        \le   2^{-\Omega(\lambda\,x_e^2/(d\,n))}
        \quad\text{and}\\
        \Pr{\Delta X_e^- > T} 
        &\le& 2^{-\Omega(T)} 
        \le   2^{-\Omega(\lambda\,x_e^2/(d\,n))}\enspace.
      \end{eqnarray*}

      Altogether,
      \begin{eqnarray}
        \nonumber
        \Pr{\Delta X_e \ge 0} &\ge& 
        \left(1 - 2^{-\Omega\left(\frac{\lambda\,x_e^2}{d\,n}\right)}\right)
        \cdot \left(1 - 2^{-\Omega\left(\frac{\lambda\,x_e^2}{d\,n}\right)}\right)\\
        &=& 1 - 2^{-\Omega\left(\frac{\lambda\,x_e^2}{d\,n}\right)} \enspace.
        \nonumber
      \end{eqnarray}
      Finally, since $x_e\ge n\,y_e/2$, $\Pr{\Delta X_e < 0} \le
      2^{-\Omega(\lambda\,n\,y_e^2/d)}=2^{-\Omega(x_e)}$. 
    \end{description}
\end{description}
In all cases, the probability that the edge becomes unused is bounded by
$2^{-\Omega(x_e)}=2^{-\Omega(n)}$. Hence, the same holds also for $m=\poly(n)$
edges and $\poly(n)$ rounds.
\end{proof} 

The proof does not only show that edges do not become empty with high
probability, but also that the congestion does not fall below any
constant congestion value. In particular, for the constant $d$ this
implies that with high probability the dynamics never reach case 2b of
the proof of Lemma~\ref{trm:virtual_potential}. This is the only place
where our analysis relies on the parameter $\nu$. Hence, for a large
number of players we can remove it from the protocol and the dynamics
converge to an exact Nash equilibrium.

\subsection{The Price of Imitation}

In the preceding section we have seen that it is unlikely that
resources become unused when the granularity of an agent decreases. If
the instance, i.\,e., the latency functions and the number of users,
is fixed, it is an interesting question, how much the performance can
suffer from the fact that the \IP is not innovative. We measure this
degradation of performance by introducing the \emph{Price of
  Imitation} which is defined as the ratio between the expected social
cost of the state to which the \IP converges, denoted $I_\Gamma$, and
the optimum social cost. The expectation is taken over the random
choices of the \IP, including random initialization.

We answer this question here for the case of linear latency functions of the
form $\lat_e(x)=a_e\,x$. Then, $d=1$ is an upper bound on the elasticity and
$\nu=a_{\max}=\max_{e\in E}\{a_e\}$. Choosing the average latency
$SC(x)=\sum_{e\in E} (x_e/n)\cdot\lat_e(x_e)$ as the social cost measure, we
show that the Price of Imitation is bounded by a constant.  It is, however,
obvious that the same also holds if we consider the makespan, i.\,e., the
maximum latency, as social cost function.

The performance of the dynamics can be artificially degraded by introducing an
extremely slow edge. Thus, $a_{\max}$ can be chosen extremely large such that
any state is imitation-stable. However, such a resource can be removed from
the instance without harming the optimal solution at all since it would not be
used anyhow. We will call such resources \emph{useless} and make this notion
precise below.

Let us first define some quantities used in the proof. For a set of
resources $M$, let $A_M= \sum_{e\in M} \frac{1}{a_e}$ and let
$A_\Gamma=A_{[m]}$. For $M\subseteq [m]$ let $\Gamma\setminus M$ denote
the instance obtained from $\Gamma$ by removing all resources in $M$. 
In the proof, we do not compare the outcome of the \IP to the optimum
solution, but rather to a lower bound, namely the optimal fractional
solution. The optimal fractional solution $\tilde x_e$ can be computed
as $\tilde x_e=n / (A_\Gamma \, a_e)$. For this solution, the latency
of all resources is $a_e\cdot \tilde x_e = n/A_\Gamma$.  A resource is
\emph{useless} if $\tilde x_e<1$. In the following, we assume that
there are no useless resources. Then, we can show that the social cost
at an imitation-stable state in which all resources are used, does not
differ by more than a small constant from the optimal social cost
(Lemma~\ref{trm:sc_bound_emitation_stable}) and that the Price of
Imitation is small. In fact, whereas $\tilde x_e\ge 1$ is required for
Lemma~\ref{trm:sc_bound_emitation_stable}, we here need a slightly
stronger assumption, namely that $x_e=\Omega(\log n)$.

\begin{theorem}
  Assume that for the optimal fractional solution, $\tilde x_e =
  \Omega(\log n)$ large enough. The price of imitation is at most
  $(3+o(1))$. In particular, for $\delta>0$, and any $n\ge
  n_0(\delta)$ for a large enough value $n_0(\delta)$ (which is
  independent of the instance),
  \[ I_\Gamma \le (3+\delta)\cdot \frac{n}{A_\Gamma}\enspace. \]
  \label{trm:price_of_imitation}
\end{theorem}

We start by proving two lemmas.

\begin{lemma}
  Let $x$ be a state in which no agent can gain more than $a_{\max}$. Then,
  \[ 
  \frac{n}{A_\Gamma} 
  \le SC(x) \le
  3\frac{n}{A_\Gamma}
  \enspace.
  \]
  \label{trm:sc_bound_emitation_stable}
\end{lemma}
\begin{proof}
  The lower bound has been proven above since $n/A_\Gamma$ is the
  social cost of an optimal fractional solution. Also note that, since
  there are no useless resources, $\tilde x_e\ge 1$ and hence
  $n/A_\Gamma\ge a_{\max}$.

  For the upper bound, consider a state $x$ in which no agent can gain
  more than $a_{\max}$. For the sake of contradiction assume that
  there exists a resource $e\in[m]$ with $\lat_e(x_e)>3\,n/A_\Gamma$.
  Since $x\neq \tilde x$ there exists a resource $f\neq e$ with $x_f <
  \tilde x_f$. In particular, $\lat_f(x_f+1) < n/A_\Gamma + a_{\max}
  \le 2\,n/A_\Gamma \le \lat_e(x_e) - a_{\max}$. The last inequality
  holds due to our assumption on $\lat_e(x_e)$ and since
  $n/A_\Gamma\ge a_{\max}$. Hence, any agent on resource $e$ can
  improve by $a_{\max}$ by migrating to $f$, a contradiction.
\end{proof}


\begin{lemma}
  \label{trm:convergence_time_emitation_stable}
  The \IP converges towards an imitation-stable state in
  time $\Oh{n^4\,\log n}$.
\end{lemma}
\begin{proof}
  Consider a state $x(t)$ in which there is at least one agent who can
  make an improvement of $a_{\max}$. Since its current latency is at
  most $n\cdot a_{\max}$ and the probability to sample the correct
  resource is at least $1/n$, the probability to do so is at least
  $\lambda\cdot (1/n)\cdot (a_{\max}/(n\,a_{\max})) = \lambda/n^2$ and
  the virtual potential gain of such a step is $a_{\max} \ge
  \Phi/n^2$. Hence, the expected virtual potential gain in state
  $x(t)$ is at least $\lambda\,\Phi(x(t))/n^4$. Hence, by
  Lemma~\ref{trm:virtual_potential},
  \[ \Ex{\Phi(x(t+1))} \le \Phi(x(t)) \cdot\left(1 - \frac{\lambda}{2\,n^4}\right)\enspace.\]
  Note that $\Phi^* \ge n\,a_{\min}$ and $a_{\max} \le n\, a_{\min}$
  by the assumption that no resource is useless.  Also, $\Phi(x(0))
  \le n^2\,a_{\max}$. Now, the theorem is an application of
  Lemma~\ref{trm:martingalelike_factor} in the appendix.
\end{proof}

Based upon the proof of Theorem~\ref{trm:empty} we can now bound the
probability that a resource becomes empty for the case of linear
latency functions more specifically.

\begin{lemma}
  \label{trm:multiple_failures}
  The probability that all resources of the subset $M\subseteq[m]$
  become empty in one round simultaneously is bounded from above by
  \[ \prod_{e\in M} 2^{-\Omega\left(\frac{n}{A_\Gamma\,a_e}\right)}\enspace.\]
\end{lemma}
\begin{proof}
  Recall the bounds on the probability that a resource $e\in[m]$
  becomes empty in the proof of Theorem~\ref{trm:empty}. Since we now
  consider linear latency functions, we may explicitly compute the
  value of $y_e=1/(A_\Gamma\,a_e)$. Recall the two cases and the
  failure probability in the initialization:
  \begin{description}
  \item[Initialization: ] Here, the error probability was at most
    $2^{-\Omega(n\,y_e)}=2^{-\Omega\left(\frac{n}{A_\Gamma\,a_e}\right)}$.
  \item[Case 1: $x_e > y_e\,n$.] Here, the error probability was at
    most $2^{-\Omega(x_e)} = 2^{-\Omega\left(\frac{n}{A_\Gamma\,a_e}\right)}$.
  \item[Case 2: $y_e\, n / 2 < x_e \le y_e\, n$.] Here, the error
    probability was at most $2^{-\Omega(x_e^2/n)} =
    2^{-\Omega\left(\frac{n}{(A_\Gamma\,a_e)^2}\right)}$.
  \end{description}
  In all cases, the probability that resource $i$ becomes empty is at
  most $2^{-\Omega\left(\frac{n}{A_\Gamma\,a_e}\right)}$.

  Furthermore, consider resources $e$ and $e'$ and let $E$ and $E'$
  denote the events that $e$ and $e'$ become empty, respectively. It
  holds that, $\Pr{E'\mid E} \le \Pr{E'}$. Therefore,
  $\Pr{E\cap E'} = \Pr{E}\cdot\Pr{E'\mid E} \le \Pr{E}\cdot\Pr{E'}$.
  Extending this argument to several resources yields the statement of
  the lemma.
\end{proof}

Using the above two lemmas, we can now prove the main theorem of this section.

\begin{proof}[Proof of Theorem~\ref{trm:price_of_imitation}]
  The proof is by induction on the number of resources $m$. Clearly,
  the statement holds for $m=1$, in which case there is only one
  assignment.  In the following we divide the sequence of state
  generated by the \IP into {\em phases} consisting of several rounds.
  The phase is terminated by one of the following events, whatever
  happens first:
  \begin{enumerate}
  \item A subset of resources $M$ becomes empty.   
  \item The \IP reaches an imitation-stable state.    
  \item The protocol enters round $\Theta(n^5\,\log n)$. 
  \end{enumerate}
  
If a phase ends because Event 1 occurs, we start a new phase for the
instance $\Gamma\setminus M$. If it ends because of Event 3, we start
a new phase for the original instance.

The probability for Event 1 is bounded by
Lemma~\ref{trm:multiple_failures}. Note that the probability is also
bounded for up to $\poly(n)$ many rounds. If a phase ends with Event 2
we have $I_\Gamma\le 3\,\frac{n}{A_\Gamma}$
(Lemma~\ref{trm:sc_bound_emitation_stable}). We bound the probability
of this event by $1$, which is trivially true.  Event 3 happens with a
probability at most $\Oh{1/n}$. This can be shown using
Lemma~\ref{trm:convergence_time_emitation_stable} and Markov's
inequality. Note that the expected social cost is still at most
$I_\Gamma$. Summing up over all three events, we obtain the following
recurrence:
  \[ I_\Gamma
  \le
  \sum_{M\subset[m]} \prod_{e\in M} 2^{-\Omega\left(\frac{n}{A_\Gamma\,a_e}\right)}\cdot I_{\Gamma\setminus M} +
  3\cdot \frac{n}{A_\Gamma} + 
  \Oh{\frac1n}\cdot I_\Gamma
  \]
  implying
  \[ I_\Gamma \cdot \left(1-\Oh{\frac1n}\right)
  \le
  3\cdot \frac{n}{A_\Gamma} + 
  \sum_{M\subset[m]} \prod_{e\in M} 2^{-\Omega\left(\frac{n}{A_\Gamma\,a_e}\right)}\cdot I_{\Gamma\setminus M}
  \enspace.
  \]  
  Substituting the induction hypothesis for $I_{\Gamma\setminus M}$,
  and introducing a constant $c$ for the constant
  in the $\Omega()$,
  \begin{eqnarray*}
    I_\Gamma \cdot \left(1-\Oh{\frac1n}\right)
    &\le&
    3\cdot \frac{n}{A_\Gamma} + 
    \sum_{M\subset[m]} \prod_{e\in M} 2^{-\frac{c\,n}{A_\Gamma\,a_e}}\cdot 4\,\frac{n}{A_{\Gamma\setminus M}}\\
    &=&
    3\cdot \frac{n}{A_\Gamma} + 
    4\,\frac{n}{A_\Gamma}\sum_{M\subset[m]} 2^{-\frac{c\,n\,A_M}{A_\Gamma}}\cdot \frac{A_\Gamma}{A_{\Gamma\setminus M}}
    \enspace.
  \end{eqnarray*}
  Now, by our assumption that for all $e\in M$, $\tilde
  x_e=n/(A_\Gamma\cdot a_e)\ge\Omega(\log n)$, we know that for all
  $e$, $1/a_e\ge c'\,A_\Gamma\cdot\log n/n$ for a constant $c'$ which
  we may choose appropriately. In particular, $A_{M}\ge |M|
  c'\,A_\Gamma\cdot\log n/n$ and $A_{\Gamma\setminus M}\ge
  c'\,A_\Gamma\cdot\log n/n$. Altogether,
  \begin{eqnarray*}
    I_\Gamma \cdot \left(1-\Oh{\frac1n}\right)
    &\le&
    \frac{n}{A_\Gamma}\left(3 +
    4\,\sum_{M\subset[m]} 2^{-c\,c'\,|M|\log n}\cdot \frac{n}{c'\,\log n}\right)\\
    &=&
    \frac{n}{A_\Gamma}\left(3 +
    4\sum_{k=1}^{m-1}{m\choose k} 2^{-c\,c'\,k\log n}\cdot \frac{n}{c'\,\log n}\right)\\
    &\le&
    \frac{n}{A_\Gamma}\left(3 +
    4\sum_{k=1}^{m-1}n^k\cdot 2^{-c\,c'\,k\log n}\cdot \frac{n}{c'\,\log n}\right)\\
    &\le&
    \frac{n}{A_\Gamma}\left(3 +
    4\sum_{k=1}^{m-1}2^{-(c\,c'-1)\,k\log n}\cdot \frac{n}{c'\,\log n}\right)\\
    &\le&
    \frac{n}{A_\Gamma}\left(3 +
    4\sum_{k=1}^{m-1} \frac{n^{-(c\,c'-1)\,k+1}}{c'\,\log n}\right)\\
    &\le&
    (3 + o(1))\,\frac{n}{A_\Gamma}
    \enspace,
  \end{eqnarray*}
  since the last sum is bounded by $o(n)$. This implies our claim.
\end{proof}



  \section{Exploring New Strategies}
\label{sec:exploration}

In Section~\ref{Sec:ExactConvergence}, we have seen that, in the long
run, the dynamics resulting from the \IP converges to an
imitation-stable state in pseudopolynomial time. The \IP and the
concept of an imitation-stable state have the drawback that the
dynamics can stabilize in a quite disadvantageous situation, e.g. when
all players play the same expensive strategy. This is due to the fact
that the strategy space is essentially restricted to the current
strategy choices of the agents. Strategies that might be attractive
and offer a large latency gain are ``lost'' once no player uses them
anymore.

A stronger result would be convergence towards a Nash equilibrium. In
the literature, several other protocols are discussed. For all of the
protocols we are aware of, the probability to migrate from one strategy
to another depends in some continuous, non-decreasing fashion on the
anticipated latency gain, and it becomes zero for zero gain. Hence, in
a setting with arbitrary latency functions which we consider here,
there always exist simple instances and states that are not at
equilibrium and in which only one improvement step is possible which
has an arbitrarily small latency gain. Hence, it takes
pseudopolynomially long, until an exact Nash equilibrium is reached.
Still, it might be desirable to design a protocol which reaches a Nash
equilibrium in the long run. There are several ways to achieve this
goal. We will discuss three of them here.

Theorem~\ref{trm:empty} states the following for a particular class of
singleton congestion games. With an increasing number of players it
becomes increasingly unlikely that useful strategies are lost. This
allows to omit the parameter $\nu$ from the protocol. If no strategies
are lost for a long period of time, the dynamics will converge towards
an exact Nash equilibrium.

Second, we may add an additional ``virtual agent'' to every strategy,
such that the probability to sample a strategy never becomes
zero. This has two implications on our analysis. On the one hand,
there is a certain base load on all resources, denoted by $x_e^0$. We
then need to have an upper bound on the elasticity of
$\lat_e(x-x_e^0)$ which may be larger than the elasticity of
$\lat_e(x)$ itself. Furthermore, we have to add $|\paths|$ virtual
agents, which leaves the analysis of the time of convergence unchanged
only if $n=\Omega(|\paths|)$.

As a third alternative, we can add an exploration component to the
protocol. With a probability of $1/2$, the agents can sample another
path uniformly at random rather than another agent. In this case,
however, the elasticity $d$ cannot be used as a damping factor anymore,
since the expected increase of congestion may be much larger than the
current load. Rather, we have to reduce the migration probability by a
factor $\min\left\{1,\frac{|\paths|\,\lat_{\min}}{\beta\,n}\right\}$
where $\beta$ is an upper bound on the maximum slope and $\ell_{\min} =
\min_{e\in E} \lat_e(1)$ is the minimum latency of an empty resource.

\begin{algorithm}[htbp]
  \caption{\XP, repeatedly executed by all players in parallel.}
  \label{alg:exploration}
  \begin{algorithmic}
    \STATE Let $P$ denote the path of the player in state $x$.
    \STATE Sample another path $Q \in \paths$ uniformly at random.
    \IF{$\lat_P(x) > \lat_Q(x+1_Q-1_P)$}
      \STATE with probability \[ \mu_{PQ} = 
      \min\left\{1,\lambda \cdot \frac{|\paths|\,\lat_{\min}}{\beta\,n}
      \cdot \frac{\lat_P(x) - \lat_Q(x+1_Q-1_P)}{\lat_P(x)}\right\}\] migrate
      from path $P$ to bin $Q$.
    \ENDIF
  \end{algorithmic}
\end{algorithm}

\begin{lemma} 
  \label{trm:Exploration:virtual_potential}
  Let $x$ denote a state and let $\Delta x$ denote a random
  migration vector generated by the \XP. Then,
  \begin{eqnarray*}
    \Ex{ \Delta\Phi(x,\Delta x) } & \le &
    \frac{1}{2} \sum_{P,Q\in\paths} \Ex{V_{PQ}(x,\Delta x)} \enspace.
  \end{eqnarray*}
\end{lemma}

\begin{proof}
  Recall that Lemma~\ref{trm:errorterm} states the following for every
  state $x$ and every migration vector $\Delta x$
  \[
  \Delta \Phi(x, \Delta x)
  \, \le \,
  \sum_{P,Q\in\paths} V_{PQ}(x, \Delta x ) + \sum_{e\in E} F_e(x, \Delta x)
  \enspace.
  \]
  Now, in order to proof Lemma~\ref{trm:Exploration:virtual_potential}, we
  apply the same approach as in the proof of
  Lemma~\ref{trm:virtual_potential}. Hence, it remains to adapt the upper
  bound on $\Ex{\Dl(\DX)}$ to the \XP. Note that this is quite simple, since
  due to the linearity of expectation,
\begin{eqnarray*}
  \Ex{\Dl(\DX)} &\le& \beta\,\Ex{\DX} \\
  &\le& \beta\,n \cdot \lambda \cdot \frac{\lat_{\min}\,|\paths|}{\beta\,n}
  \cdot \frac{1}{|\paths|} \cdot \frac{\lat_{P} - \lat_Q^+}{\lat_P}\\ &\le& \lambda
   \cdot \frac{\lat^+_{e}}{\lat^+_Q}\cdot (\lat_{P} - \lat_Q^+)\enspace,
  \end{eqnarray*}
  where we have substituted the migration probability of the protocol
  and the fact that there are at most $n$ agents that may sample a
  path containing $e$. This proves Equation~(\ref{eqn:tilde_l_1}) if
  $\lambda$ is chosen small enough. With opposite signs, the same
  argument holds if $e\in P$.
\end{proof}

Since we have omitted the parameter $\nu$ from the protocol, we now
need a lower bound on the minimum improvement that is possible when
the system is not yet at an imitation-stable state in order to give an
upper bound on the convergence time. Formally, let
\[ \kappa = \min_{x}\min_{\begin{array}{c}P,Q\in\paths\\\lat_p(x) >
 \lat_Q(x+1_Q-1_P)\end{array}} \{\lat_P(x) - \lat_Q(x+1_Q-1_P)\}\enspace. 
\]

\begin{theorem}
  Consider a symmetric network congestion game in which all players
  use the \XP. Let $x$ denote the initial state of the dynamics. Then the
  dynamics converge to a Nash equilibrium in expected time
  \[ \Oh{\frac{\Phi(x)\,\beta\,n\,\lat_{\max}}{\lat_{\min}\,\kappa^2}}\enspace. \]
\end{theorem}
\begin{proof}
  In every state which is not a Nash equilibrium there exists an agent
  currently utilizing path $P\in\paths$ and a path $Q\in\paths$ such
  that $\lat_Q \le \lat_P-\kappa$. Hence, the expected virtual
  potential gain is at least
  \[ \Ex{V_{PQ}}
  \le -\frac{1}{|\paths|}\cdot
  \frac{\lambda\,|\paths|\,\lat_{\min}}{\beta\,n}\cdot\frac{\kappa}{\lat_P}\cdot\kappa \le -\frac{\lambda\,\lat_{\min}}{\beta\,n}\cdot\frac{\kappa^2}{\lat_{\max}}
  \enspace,
  \]
  and the true potential gain is at least half of this.  Again,
  Lemma~\ref{trm:martingalelike} yields the expected time until the
  potential decreases from at most $\Phi$ to $\Phi^*\ge 0$.
\end{proof}

It is obvious that an analogue of Lemmas~\ref{trm:virtual_potential}
and~\ref{trm:Exploration:virtual_potential} also holds for any
protocol that is a combination of the \IP and the \XP, e.\,g., a
protocol in which in every round, every agent executes the one or the
other with probability one half. Then, in order to bound the value of
$\Ex{\Dl(\DX)}$, we must make a case differentiation based on whether
proportional or uniform sampling dominates the probability that other
agents migrate towards resource $e$. Such a protocol combines the
advantages of the \IP and the \XP: In the long run, it converges to a
Nash equilibrium, and reaches an approximate equilibrium as quickly as
stated by Theorem~\ref{trm:convergence_bicriteria} (up to a factor of
$2$).


  \section{Conclusion}

We have proposed and analyzed a natural protocol based on imitating
profitable strategies for distributed selfish agents in symmetric
congestion games.  If agents use our \IP, the resulting dynamics
converge rapidly to approximate equilibria, in which only a small
fraction of players have latency significantly above or below the
average. In addition, in finite time the dynamics converges to an
imitation-stable state, in which no player can improve its latency by
more than $\nu$ by imitating a different player.  The \IP and the
concept of an imitation-stable state have the drawback that dynamics
can stabilize in a quite disadvantegous situation, e.g. when all
players play the same expensive strategy. This is due to the fact that
the strategy space is essentially restricted to the current strategy
choices of the agents. Strategies that might be attractive and offer
large latency gain are ``lost'' once no player uses them anymore. For
singleton congestion games we showed that this event becomes unlikely
to occur as the number of players increases. Then, by removing
parameter $\nu$ from the protocol, the dynamics become likely to
converge to Nash equilibria. Another approach to avoid losing
strategies is to include exploration of the strategy space. Towards
this end, we can use an \XP, in which players sample from the strategy
space directly and then migrate with a certain probability. If every
player uses a suitably designed \XP (or any random combination of \XP
and \IP), then the dynamics are always guaranteed to converge to a Nash
equilibrium. However, acquiring information about possible strategies
and their benefits might be a complex and costly process in practice,
and hence such an action should be invoked only rarely. In addition,
exploration requires small migration probabilities, because the danger
of overshooting is more severe. Thus, on the downside, if the \XP is
used exclusively, this results in significantly larger convergence
times.


\bibliographystyle{plain}
\bibliography{strings,all_references,simon}

\clearpage
 
\begin{appendix}
  \section{Appendix}

\subsection{Useful Facts}

Throughout the technical part of this paper, we will apply the following two 
Chernoff bounds. \begin{fact}[Chernoff, see~\cite{Hagerup/Rueb:GuidedTour:90}]
  \label{trm:chernoff}
Let $X$ be a sum of Bernoulli variables. Then, $
  \Pr{X \ge k\cdot \Ex{X}}  
\le \expf{-\Ex{X}\,k\cdot(\ln k - 1)} $, and, for $k\ge 4 > \expf{4/3}$, $
  \Pr{X \ge k\cdot \Ex{X}}  
\le \expf{-\frac14\,\Ex{X}\,k\,\ln k} $. Equivalently, for $k\ge 4\,\Ex{X}$, $
  \Pr{X \ge k}
\le \expf{-\frac14\,k\,\ln (k/\Ex{X})} $.
\end{fact}

\medskip

The following fact yields a linear approximation of the exponential function.
\begin{fact}
  \label{trm:exp_lin}
For any $r>0$ and $x\in [0,r]$, it holds that $(\expf{x} - 1) \le x\cdot 
\frac{\expf{r}-1}{r}$.
\end{fact}
\begin{proof}
The function $\exp(x)-1$ is convex and it goes through the points $(0,0)$ and 
$(r,\expf{r}-1)$, as does the function $x\cdot \frac{\expf{r}-1}{r}$.
\end{proof}

\medskip

\begin{fact} \label{trm:geomSum}
  For every $c \in ]0,1[$ it holds
  \begin{eqnarray*}
  \sum_{k=0}^{\infty} c^k &=& \frac{c}{1-c} \\
  \sum_{k=l}^{\infty} c^k &=& \frac{c^l}{1-c} \\
  \end{eqnarray*}

\end{fact}

\medskip

\begin{fact}[Jensen's Inequality]
  \label{trm:jensen}
Let $f \colon \RR \rightarrow \RR$ be a convex function, and let
$a_1,\ldots,a_k,x_1,\ldots,x_k \in \RR$. Then
\[\begin{array}{crcl}
  & \displaystyle f \left( \frac{\sum_{i=1}^k a_i x_i}{\sum_{i=1}^k a_i} \right)
  & \le &
    \displaystyle \frac{\sum_{i=1}^k a_i f(x_i)}{\sum_{i=1}^k a_i} \enspace.
\end{array}\]

If $f(x) = x^2$, then    
\[\begin{array}{crcl}
  & \displaystyle \left( \frac{\sum_{i=1}^k a_i x_i}{\sum_{i=1}^k a_i} \right)^2
  & \le &
    \displaystyle \frac{\sum_{i=1}^k a_i (x_i)^2}{\sum_{i=1}^k a_i} \\[5ex]
  \Leftrightarrow
  & \displaystyle \frac{1}{\sum_{i=1}^k a_i} \cdot \left(\sum_{i=1}^k
  a_i x_i \right)^2
  & \le &
    \displaystyle \sum_{i=1}^k a_i f(x_i) \enspace.	
\end{array}\]
\end{fact}

\medskip

\begin{lemma}[\cite{Fischer/etal:WardropFinite:DC:08}]
  \label{trm:martingalelike}
  Let $X_0,X_1,\ldots$ denote a sequence of non-negative 
  random variables and assume that for all $i\ge 0$
  \[ \Ex{X_i \mid X_{i-1}=x_{i-1}} \le x_{i-1} - 1 \]
  and let $\tau$ denote the first time $t$ such that $X_t=0$. Then,
  \[ \Ex{\tau\mid X_0=x_0} \le x_0\enspace. \]
\end{lemma}

\medskip

\begin{lemma}[\cite{Fischer/etal:WardropFinite:DC:08}]
  \label{trm:martingalelike_factor}
  Let $X_0,X_1,\ldots$ denote a sequence of non-negative random
  variables and assume that for all $i\ge 0$
  $\Ex{X_i \mid X_{i-1} = x_{i-1}} \le x_{i-1} \cdot \alpha$
  for some constant $\alpha\in(0,1)$. Furthermore, fix some constant
  $x^*\in(0,x_0]$ and let $\tau$ be the random variable that describes
  the smallest $t$ such that $X_t \le x^*$. Then,
  \[
  \Ex{\tau\mid X_0=x_0}
  \le \frac{2}{\log(1/\alpha)}\cdot \log\left(\frac{x_0}{x^*}\right) 
  \enspace.
  \]
  Again, as a consequence of Lemma~\ref{trm:martingalelike} the
  expected time until the potential decreases from at most $\Phi$ to
  $\Phi$ can be found in the appendix, and which is proved, e.\,g.,
  in~\cite{Fischer/etal:WardropFinite:DC:08}.
\end{lemma}

\end{appendix}

\end{document}